\documentclass{article}
\usepackage{braket, amsmath,amssymb,amsfonts}
\usepackage{array}
\usepackage{graphicx}
\usepackage{textcomp}
\usepackage{times}
\usepackage{epsfig}
\usepackage{amsmath}
\usepackage[numbers]{natbib}
\usepackage{algorithmic}
\usepackage{algorithm}
\usepackage{multirow}
\usepackage[flushleft]{threeparttable}
\usepackage{wrapfig}
\usepackage[pagebackref=false,breaklinks=true,colorlinks,bookmarks=false]{hyperref}
\usepackage{amsthm}

\usepackage{bbm}
\usepackage{booktabs}
\usepackage[total={6in,8in}]{geometry}

\newtheorem{remark}{Remark}

\newtheorem{condition}{Condition}
\newtheorem{proposition}{Proposition}
\newtheorem{theorem}{Theorem}
\newtheorem{lemma}{Lemma}

\usepackage{amsopn}

\DeclareMathOperator*{\argmin}{arg\,min}
\newcommand{\goodchi}{\protect\raisebox{2pt}{\large$\chi$}}

\def\BibTeX{{\rm B\kern-.05em{\sc i\kern-.025em b}\kern-.08em
    T\kern-.1667em\lower.7ex\hbox{E}\kern-.125emX}}

\providecommand{\keywords}[1]
{
  \small	
  \textbf{\textit{Keywords---}} #1
}
    
\begin{document}    

\title{A Deep Neural Network Algorithm for Linear-Quadratic Portfolio Optimization with MGARCH and Small Transaction Costs}

\author{Andrew Papanicolaou\thanks{Department of Mathematics, North Carolina State University, Campus Box 8205, Raleigh, NC 27695. (apapani@ncsu.edu). Corresponding author. This work was partially supported by NSF grant DMS-1907518.} Hao Fu\thanks{Electrical and Computer Engineering, NYU Tandon School of Engineering, Brooklyn, NY 11201 (hf881@nyu.edu)}, Prashanth Krishnamurthy\thanks{Electrical and Computer Engineering, NYU Tandon School of Engineering, Brooklyn, NY 11201 (prashanth.krishnamurthy@nyu.edu)}, Farshad Khorrami\thanks{Electrical and Computer Engineering, NYU Tandon School of Engineering, Brooklyn, NY 11201 (khorrami@nyu.edu)}}

\date{}
\maketitle
\begin{abstract}
We analyze a fixed-point algorithm for reinforcement learning (RL) of optimal portfolio mean-variance preferences in the setting of multivariate generalized autoregressive conditional-heteroskedasticity (MGARCH) with a small penalty on trading. A numerical solution is obtained using a neural network (NN) architecture within a recursive RL loop. A fixed-point theorem proves that NN approximation error has a big-oh bound that we can reduce by increasing the number of NN parameters. The functional form of the trading penalty has a parameter $\epsilon>0$ that controls the magnitude of transaction costs. When $\epsilon$ is small, we can implement an NN algorithm based on the expansion of the solution in powers of $\epsilon$. This expansion has a base term equal to a myopic solution with an explicit form, and a first-order correction term that we compute in the RL loop. Our expansion-based algorithm is stable, allows for fast computation, and outputs a solution that shows positive testing performance.
\end{abstract}


\tableofcontents

\section{Introduction}

The problem we consider is one faced by a fund manager who has just taken in a large amount of new capital. This capital needs to be integrated into the portfolio but transaction costs caused by large bid/ask spreads make it extremely inefficient to directly invest the entire amount immediately (i.e., the typical buy-and-hold strategy is sub-optimal).  A more efficient way is to invest the new funds according to a solution to a dynamic mean-variance optimization that includes a quadratic penalty on trade size. Optimal execution of large orders was formulated as a mean-variance optimization with penalization on trades in \cite{almgren2001optimal}, and a multi-asset version of this problem was studied in \cite{garleanu2013dynamic}. However in practice, many assets have heteroskedasticity, and therefore it is interesting to consider mean-variance preferences in the setting of dynamic covariance matrices given by a multi-variate GARCH (MGARCH) model. In addition, when volatility spikes, there is an increase in price impact (\cite{capponi2019trade}). This negative correlation between price and volatility was also found by \cite{black1976studies}. \cite{mantalos2020improved} modeled such volatility by fitting skewness in the ARCH model. This work assumes that  the penalty on trading depends on the instantaneous value of the covariance matrix (e.g., the condition number of the covariance matrix) to describe such price volatility. The contrary movement between degrees of freedom in covariance matrices and overall market volatility is highlighted in \cite{avellaneda2010statistical} and also touched upon in \cite{laloux2000random}. This heteroskedastic problem is a linear-quadratic program, but with the added feature of non-constant coefficients that depend on the MGARCH process. 

When solving this linear-quadratic program with non-constant coefficients in high dimensions, there are substantial computational challenges for non-ML approaches. There is no explicit solution because of non-constant coefficients; this is in contrast with a linear-quadratic program with constant coefficients for which the solution has a feedback form given by a matrix Riccati equation. For non-constant coefficients and high dimensions, the Riccati equation is not explicitly solvable, but a more tractable solution uses a neural network (NN) to obtain a policy approximation that is refined through reinforcement learning (RL). This estimated policy is the output of an iterated algorithm that is similar to the actor-critic approach used in the deep-deterministic policy gradient (DDPG \cite{2013arXiv1312.5602M}) and the deep Q Network (DQN \cite{2015arXiv150902971L}) for discrete action spaces.

The NN-based policy approximations have the capacity to manage high dimensions, but  the implementation of the method requires some analysis of the problem. In particular, actor-critic RL algorithms are known to be unstable in many situations (\cite{han2020actor,yang2019provably}). Stability and computation time for RL algorithms are discussed in \cite{wang2018boosting} along with a boosting algorithm to mitigate these issues. In \cite{van2015deep}, a double-network method is used to stabilize the DQN method; a double network is used in \cite{2015arXiv150902971L} wherein it is referred to as the actor-critic approach; the normalized advantage function in \cite{wang2016dueling} uses dueling networks.

In this paper, we resolve these stability and computation-time issues by constructing a solution that exploits the smallness of the transaction costs and uses the explicit form of the solution to a myopic problem. We call a solution myopic when its decision does not forecast the future and is purely based on historical observations, that is, the optimization has a discount parameter $\delta\in[0,1)$ and for the case $\delta=0$ the solution is myopic. The optimization also has transaction costs with a small parameter $\epsilon>0$, and as $\epsilon$ tends toward zero the solution becomes more and more like the myopic solution. In this paper, we will show that the explicit myopic formula is the base term in an expansion of the optimal solution in powers of $\epsilon$. The principal effects of non-myopic RL are captured by the order-$\epsilon$ term in this expansion, which means that a good approximation is the base term plus the order-$\epsilon$ term.

\subsection{Background and Review}

Other machine learning applications in finance include \cite{sirignano2018dgm} where stochastic gradient descent (SGD) with deep NN architecture is used for computing prices of American options on large baskets of stocks, and in \cite{han2018deepBSDE} where an RL approach is used to numerically solve high-dimensional backward stochastic differential equations related to finance. In \cite{fischer2018deep}, the authors utilized an LSTM network for predicting the price movement with daily S\&P500 data\footnote{https://en.wikipedia.org/wiki/S\%26P\_500\#cite\_note-history-17}. The performance of the LSTM is mixed during different periods. The short-term trend prediction in the price movement on NASDAQ by the deep network was studied by \cite{namdari2021multilayer}. The authors utilized features from both fundamental and technical analysis as the network input. \cite{kim2019hats} used a graph network to predict the stock price movement on S\&P500 data.   In \cite{liang2018adversarial}, they demonstrate how adversarial learning methods can be used to automate trading in stocks. General methods from control theory have been applied for optimal trading decisions in \cite{BP15,MPB18}. The effects of transaction costs and liquidity are well-studied (\cite{almgren2001optimal,chandra2019singular,rogers2010cost}). In particular, the ``aim portfolio" description given in \cite{garleanu2013dynamic} has been a key result for the management of large funds. The discussed works are based on supervised learning. Therefore, the non-supervised learning approaches should also be studied as they may address more complicated problems.

It was originally thought that combining NN with RL recursions would lead to instability (\cite{yang2019provably}) due to the approximation errors caused by using the neural networks. However, hypothetically, the accuracy of NN approximations can be improved to within arbitrarily close bounds (see \cite{cybenko1989approximation,pinkus1999}). Additionally, the actor-critic approach has been shown in settings to have a stabilizing effect (\cite{mnih2016asynchronous,sutton2000policy}).  The actor-critic approach utilizes two independent deep networks to represent a policy function (the actor) that provides a group of possible actions for a given state and an evaluation function (the critic) that evaluates the action taken by the actor based on the current policy function. By alternatively updating the actor and critic with the given objective function, the two networks converge. The DDPG algorithm (\cite{2015arXiv150902971L}) takes the actor-critic approach for solving RL problems with continuous control space and is based on the DPG algorithm in \cite{Silver:2014:DPG:3044805.3044850}. The DQN method in \cite{2013arXiv1312.5602M} is an RL algorithm that utilizes and trains a deep network to represent the Q-value function. Then at each instant, the DQN takes the action based on the Q-value network. The DQN does not use an actor-critic approach. and can only address discrete action problems, whereas the DDPG, which extends from DQN, can also address continuous action problems.  But the DDPG is more difficult to train and sometimes unstable as it uses an actor-critic approach. Reinforcement learning was utilized in many topics, such as wireless (\cite{chen2021rdrl}) and mobile edge computing networks (\cite{chen2022game,chen2022gpds}).

 In portfolio management, \cite{jiang2017deep} utilized the DDPG to optimize the cryptocurrency portfolio. \cite{liang2018adversarial} utilized both DDPG and proximal policy optimization (PPO) for portfolio management. \cite{buehler2019deep} considers portfolio optimization for a transaction cost. However, the authors only considered a linear cost. In this work, we considered a quadratic transaction cost that is more practical. Similar to the problem we consider in this paper is the linear-quadratic regulator with uncertainty in the (constant) coefficients. There are provable bounds for identifying the minimum run-time required from an on-policy approximation, after which the observed data is used to refine the policy to within a given tolerance of the optimal (\cite{dean2019sample,krauth2019finite}). This approach to uncertainty in on-policy learning is analyzed as an actor-critic approach in \cite{yang2019provably}.

\subsection{Results in this Paper}
The analyses in this paper show the effectiveness of RL and NN-based policy approximations for solving linear-quadratic programs with non-constant coefficients and small penalization on control. We consider an iterative scheme for the optimal controls, for which  we can prove convergence to a fixed point under some reasonable conditions. We extend this fixed point argument to show that NN approximation errors will compound over time, but that their total will be of order big-oh in the magnitude of 
an error that is reduced as the number of NN parameters is increased. The problem is of practical interest because the MGARCH covariance process does not allow for an explicit solution, such as those seen in other linear-quadratic optimizations (\cite{chow1975analysis,recht2019tour}). 

For a faster implementation, we propose an algorithm that exploits the smallness of the transaction cost parameter. In particular, we write the solution as a series expansion in powers of the transaction costs parameter, which is small. Using only a few terms from this expansion we form a sub-optimal solution that tends toward the optimum as the parameter decreases toward zero. This approach is similar to a standard NN-based policy gradient but is more stable (i.e., the difference between our RL output and the ground-truth optimal output is bounded) because the series terms are fast and easy to compute, and because the NN approximation error is present only in the higher-order terms and therefore is an order of magnitude smaller than it would've been without the expansion. In our experiments, we  take the first two terms in this series: the based term that is equal to an explicitly computable myopic solution and a first-order correction term that we compute using an RL loop and NN functional approximation. We implement this approximation using a single network that contains only fully-connected layers without requiring dueling networks or special machinery, and in our studies on market data, we can see that the correction term provides some improvement compared to using only the myopic control, which is defined in \eqref{eq:myopic}. The key difference between our RL strategy and the myopic strategy is that the myopic strategy makes decisions without forecasting the future, whereas the RL strategy forecasts the future to make decisions.
Overall, the contribution of our paper includes the following:
\begin{itemize}
    \item We show the effectiveness of RL and NN-based policy approximations for solving linear-quadratic programs with non-constant coefficients and small penalization on control.
    \item We prove the convergence of an iterative scheme to a fixed point for the optimal controls under some reasonable conditions.
    \item We show that NN approximation errors will be of order big-oh bound.
    \item We propose an algorithm based on NN approximation that exploits the smallness of the transaction cost parameter for faster implementation.
    \item We evaluate our algorithm on both synthetic and historical market data, which shows positive testing results.
\end{itemize}

The remaining part of this paper is organized as follows: First, the problem is mathematically formulated. A solution with a two-step iteration scheme is proposed, and its convergence is analyzed. A practical method for implementing the solution using neural networks is proposed. A small-$\epsilon$ analysis regarding the neural network solution is shown. Lastly, the method is evaluated on synthetic market data and historical market data, and the results are discussed.

\section{Model and Optimization Problem}
\label{sec2}
Let $R_t\in\mathbb R^n$ denote the vector of $n$-many assets' returns realized at time $t$, and let $\Sigma_{t-1}$ be the covariance of $R_t$ given the information immediately prior at time $t-1$. An MGARCH model (\cite{engle1995,Bauwens2006}) is the following:
\begin{align}
\label{eq:returns}
R_{t+1}& = \mu  + Z_{t+1}\\
\label{eq:volMatrix}
\Sigma_{t+1}&= CC^\top +A\Sigma_tA^\top + BZ_{t+1}Z_{t+1}^\top B^\top\ 
\end{align}
where $\mu$ is the conditional mean vector, $Z_{t+1}\sim \hbox{iid}(0,\Sigma_t)$, and where $A$, $B$ and $C$ are $n\times n $ matrices. The following condition ensures that $\Sigma_t$ is always invertible:
\begin{condition}
    \label{cond:C}
    The matrix $C$ given in \eqref{eq:volMatrix} is full rank so that $CC^\top$ is invertible, that is, there is a positive lower bound $c = \inf_{\|v\|=1}v^\top CC^\top v>0$.
\end{condition}
 Asset prices are calculated by compounding the returns \eqref{eq:returns}. For $1\leq i\leq n$, the $i^{th}$ asset's price is
\begin{equation}
    \label{eq:S}
    S_{t+1}^i=h(S_t^i,R_{t+1}^i)
\end{equation}
where $h$ is some known function. A typical choice of $h$ is $h(s,r) = s(1+r)$ as used in \cite{engle1995,Bauwens2006}, but for technical reasons, we will need to impose the following condition: 
\begin{condition}
    \label{cond:S_LB}
    The function $h$ in \eqref{eq:S} is finitely bounded away from zero. That is, $\|h\|_\infty =\sup_{s,r}|h(s,r)|<\infty$ and there exists constant $\underbar s>0$ such that $\inf_{s,r}h(s,r)\geq \underbar s$. 
\end{condition}

Denote the $\mathbb R^n$ vector of these prices as
\[\vec S_t = 
\begin{pmatrix}
    S_t^1\\S_t^2\\\vdots\\S_t^n
\end{pmatrix}\ .
\]
Next, define the covariance matrix of the dollar returns,
\begin{align*}
    P_t&= \Psi_t\Sigma_t\Psi_t
\end{align*}
where 
\begin{align*}
    \Psi_t = \hbox{diag}(\vec S_t)\ .
\end{align*}
Let $X_t\in\mathbb R^n$ denote a manager's holdings in assets (in contract units). The returns (in dollar units) on this portfolio are $\sum_i(S_{t+1}^i-S_t^i)X_t^i$, the expected value of these returns is $\mu^\top S_tX_t$, and their variance is $X_t^\top P_tX_t$. The portfolio manager has a control $\{a_t, t=1,2,3,\dots \}$ that she selects at time $t$ to change $X_t$. The manager's control should be optimal with respect to her mean-variance preferences, 
\begin{align}
\label{eq:valueFunction}
&V(x,s,p)=  \sup_a\sum_{t=1}^\infty \delta^t \mathbb E\left[f(a_t,X_t,\vec S_t,P_t)\Big| X_0=x,\vec S_1=s,P_1=p\right]\\
\nonumber
&~~~~\hbox{s.t.}\\
\nonumber
&f(a_t,X_t,\vec S_t,P_t)=\underbrace{\mu^\top \Psi_t X_t-\frac{\epsilon q(\vec S_t,P_t)}{2}a_t^\top \Psi_t a_t}_{\hbox{expected return}} -\underbrace{\frac{\gamma}{2}X_t^\top P_t X_t}_{\hbox{risk penalty}}\\
\nonumber
&X_t=X_{t-1} + a_t\  
\end{align}
where $0\leq\delta<1$ is a discount factor, $\epsilon>0$ is a (small) parameter, and the function $q:\mathbb R^{ n}\times \mathbb R^{n\times n}\rightarrow \mathbb R^+$ is a transaction cost (or liquidity penalty) that is higher at times when it is harder to trade and lower at times when there is plenty of liquidity. This form of transaction cost penalty was introduced in \cite{almgren2001optimal} and the relationship with volatility was shown in \cite{capponi2019trade}.
\begin{condition}
    \label{cond:q_bound}
    We assume $q$ is bounded away from zero,
    \begin{equation}
        \label{eq:qBound}
        1\leq q(s,p)\leq \goodchi<\infty\qquad\forall (s,p)\in\mathbb R^n\times\mathbb R^{n\times n} \,
    \end{equation}
where $\goodchi$ is a known constant. 
\end{condition}
\noindent
As mentioned earlier, the penalty on trading should depend on the instantaneous value of the covariance matrix. Many works use principal component analysis (PCA) to study the relationship between degrees of freedom in the stock returns' covariance matrix and overall market volatility. For example, \cite{avellaneda2010statistical} observes that relatively few eigenvectors are needed to capture the majority of market variance during times of high market stress, thus resulting in wider bid/ask spreads and higher transaction costs; the relationship is reversed during times of low market stress. Based on this dynamic, we use the condition number to excite the transaction cost function when the market has reduced degrees of freedom. In the examples, we take $q(s,p) = \hbox{cond}(\hbox{diag}^{-1}(s)p\hbox{diag}^{-1}(s))$, which is based on the empirical observation that losses in the S\&P500 market index occur when there is a large spike in the condition number of the covariance matrix, as shown in Fig.~\ref{condition}.  Additionally, we take $\gamma=\frac{1}{W_0}\sum_i(\overline\Sigma^{-1}\mu)^i$ where $\overline{\Sigma}$ is the covariance matrix estimated from market data and $W_0$ is the value of the new capital, i.e., we set risk aversion so that the goal is for $W_0$ to be invested in the market.

\newcommand\scales{0.2}
\begin{figure}[ht!]
        \centering
        \includegraphics[scale=\scales]{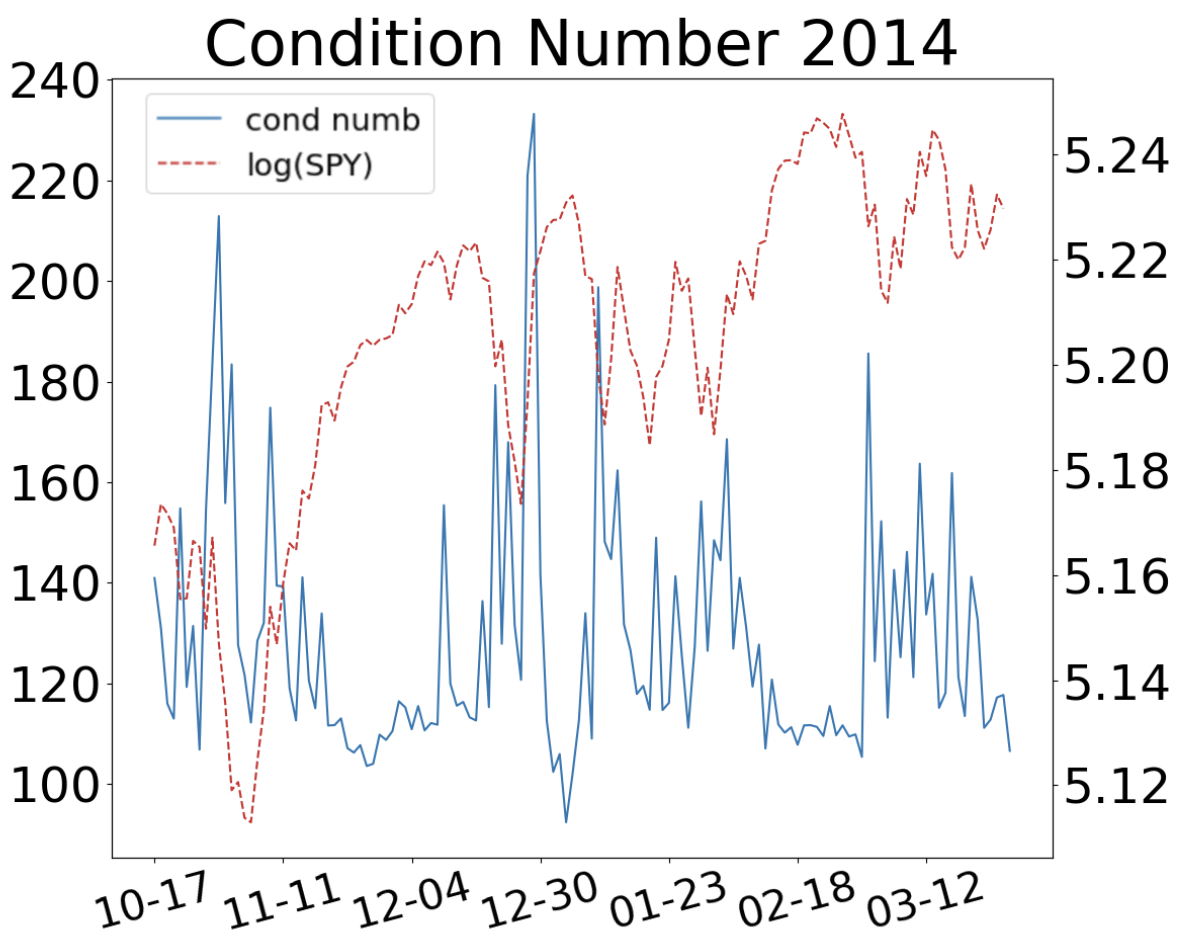}
        \hspace*{-0.10in}
        \includegraphics[scale=\scales]{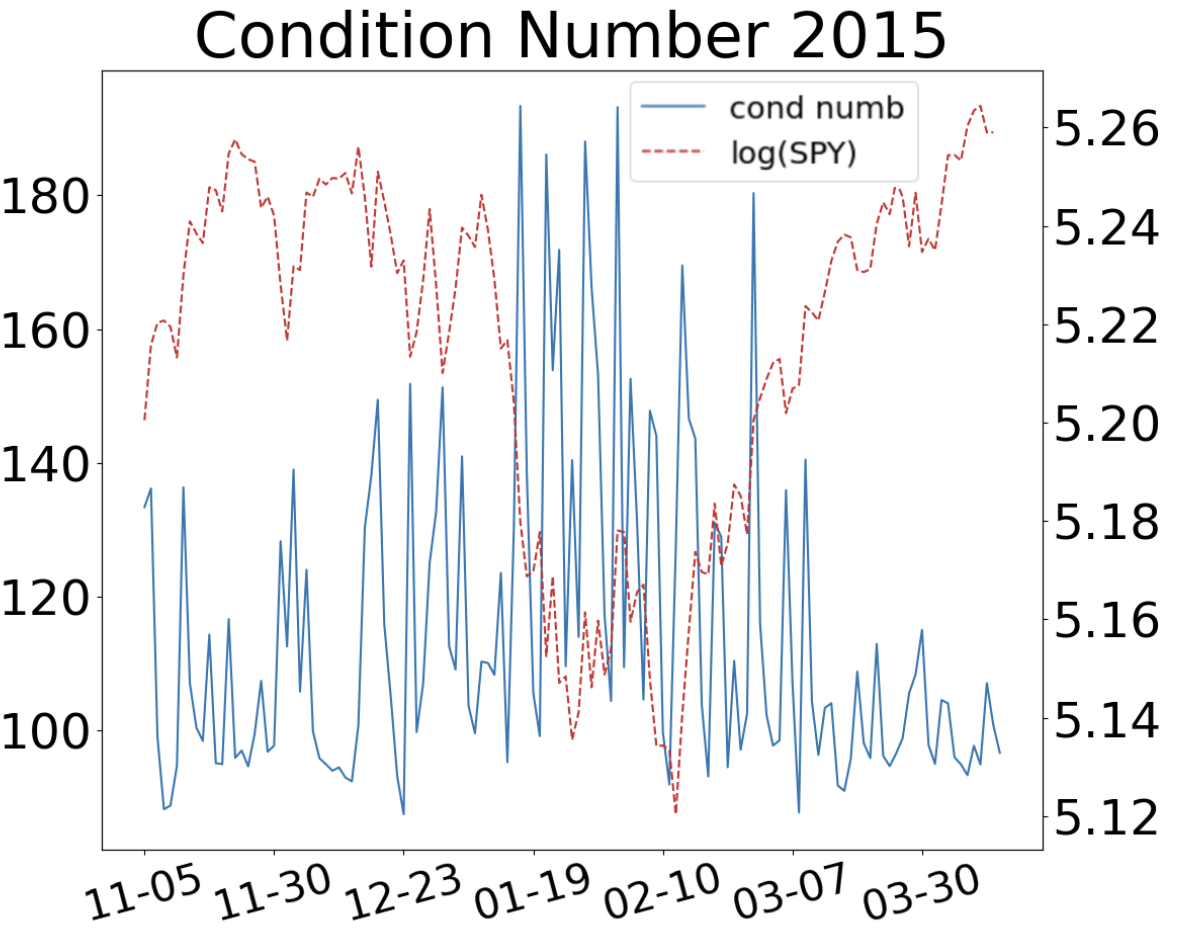}
        \caption{X-axis: the date (month - day).  Left Y-axis: condition number. Right Y-axis: log(SPY) where SPY denotes the S\&P500 index ETF. Blue plot: condition number of the covariance matrix with time. Red plot: log of SPY with time.}
        \label{condition}
\end{figure}

The problem formulated in \eqref{eq:valueFunction} is similar to the mean-variance preferences problem in \cite{garleanu2013dynamic} but with the added non-constant cost from $q(\sigma)$, where $\sigma^2 = \Sigma$ and $\Sigma$ is the estimated covariance matrix. A similar type of financial control problem was considered in \cite{malekpour2013stock}. An effective way to analyze this system is to take a Hamiltonian approach and write it using a vector of Lagrange multipliers,
\begin{align*}
	&\sum_{t=1}^\infty \delta^t\mathbb E\Big[f(a_t,X_t,\vec S_t,P_t)+\lambda_t^\top(X_t-X_{t-1}-a_t)\Big]\\
	&=\sum_{t=1}^\infty \delta^t\mathbb E\Big[f(a_t,X_t,\vec S_t,P_t)-(\delta\lambda_{t+1}-\lambda_t)^\top X_t-\lambda_t^\top a_t\Big]  +\lim_{t\rightarrow\infty}\mathbb E[\delta^{t+1} \lambda_{t+1}^\top X_t] - \lambda_1^\top X_0 \ 
\end{align*}
where we have used the transversality condition (\cite{leonard1992optimal})
$$\lim_{t\rightarrow\infty}\mathbb E(\delta^{t+1} \lambda_{t+1}^\top X_t)  = 0.$$ First-order conditions in $a_t$ and in $X_t$ yield a forward-backward system
\begin{align}
\label{eq:Xequation}
X_t& =X_{t-1} -\frac{1}{\epsilon q(\vec S_t,P_t)}\Psi_t^{-1}\lambda_t\\
\label{eq:lambdaEquation}
\lambda_t&= \delta\mathbb E_t \lambda_{t+1} - \Psi_t\mu+\gamma P_tX_t\ ;
\end{align}
where $\mathbb E_t$ denotes expectation conditional on the information observed up to time $t$. 
However, in the real world, \eqref{eq:Xequation} and \eqref{eq:lambdaEquation} cannot be directly solved. Therefore, in this paper, we propose an iteration scheme for \eqref{eq:Xequation} and \eqref{eq:lambdaEquation}:
\begin{align}
    \nonumber
    X_t^{(k+1)}&= X_{t-1}^{(k+1)}-\frac{1}{\epsilon q(\vec S_t,P_t)}\Psi_t^{-1}\lambda_t^{(k)}\\
    \label{eq:iterationScheme}
    \lambda_t^{(k+1)}&=\left(I+\widetilde P_t\Psi_t^{-1}\right)^{-1}\left(\delta\mathbb E_t\lambda_{t+1}^{(k+1)} -\Psi_t\mu+\gamma P_tX_{t-1}^{(k+1)}\right)
\end{align}
where $\widetilde P_t = \frac{\gamma}{\epsilon q(\vec S_t,P_t)}P_t$ and $X_0^{(k)} = x_0$ for all rounds of iteration $k$. Then, we implement RL using neural networks (NNs) to estimate the limiting fixed point from \eqref{eq:iterationScheme}. The convergence of iterations in \eqref{eq:iterationScheme} depends on if the following condition holds:
\begin{condition}
    \label{cond:damping}
    There is a constant $\Delta(\epsilon)<1$ such that
    \begin{equation*}
    \delta\left\|\left(I+\widetilde P_t\Psi_t^{-1}\right)^{-1}\right\|\leq \Delta(\epsilon)\ , 
\end{equation*}
for all $t$.
\end{condition}

If Condition \ref{cond:damping} holds, then we can prove that \eqref{eq:iterationScheme} converges to a unique fixed point. Condition \ref{cond:damping} would always hold if $P_t$ and $\Psi_t$ commute, but this is an unrealistic condition. For the data used in this paper, we check empirically that in fact Condition \ref{cond:damping} holds for $\epsilon<1$. For theoretical purposes, because we are considering the small-$\epsilon$ parameterization, the following proposition is useful for confirming Condition \ref{cond:damping}:

\begin{proposition}
    \label{prop:deltaEps}
    Assume Condition \ref{cond:C}, Condition \ref{cond:S_LB}  and Condition \ref{cond:q_bound}. If $\epsilon \goodchi\|h\|_\infty<\gamma\underbar s^2c$, then denoting $\kappa = \goodchi\|h\|_\infty/(\gamma\underbar s^2c)$, we have       \begin{equation}
            \label{eq:smallEpsBound}
            \sup_t\left\|\left(I+\widetilde P_t\Psi_t^{-1}\right)^{-1}\right\|\leq  \frac{\epsilon \kappa}{1-\epsilon^2\kappa^2}\ ,
        \end{equation}
    from which it follows that Condition \ref{cond:damping} will hold for $\epsilon$ small enough.
\end{proposition}
\begin{proof}(see Appendix).
\end{proof}
\subsection{Two-Step Iteration Scheme}
\label{sec:2stepPointScheme}
The scheme in \eqref{eq:iterationScheme} is the basis for an iterative algorithm with two steps. The policy is given by $\lambda^{(k)}$ and is used to generate $X^{(k+1)}$. Then, upon observation of $X^{(k+1)}$, the updated Lagrange multiplier $\lambda^{(k+1)}$ is obtained via a fixed-point iteration,
\begin{align}
    \label{eq:lambda_kPrime}
    &\lambda_t^{(k'+1)}=\left(I+\widetilde P_t\Psi_t^{-1}\right)^{-1}\left(\delta\mathbb E_t\lambda_{t+1}^{(k')} -\Psi_t\mu+\gamma P_tX_{t-1}^{(k+1)}\right)\ ,
\end{align}
for fixed $k$ and for $k'\rightarrow \infty$. Algorithmically, this can be described as a 2-step iterative procedure for finding a fixed point:
\begin{enumerate}
    \item $X_t^{(k+1)} = X_{t-1}^{(k+1)}-\Psi_t^{-1}\lambda_t^{(k)}/(\epsilon q(\vec S_t,P_t))$,
    \item $\lambda_t^{(k+1)} = \lim_{k'\rightarrow\infty}\lambda_t^{(k')}$ where $\lambda_t^{(k')}$ is given by \eqref{eq:lambda_kPrime}.
\end{enumerate}
 If we consider a finite-time version of \eqref{eq:Xequation} and \eqref{eq:lambdaEquation} with terminal condition $\lambda_{T+1}^{(k)} \equiv 0$ for all $k$ (i.e., on $X_t^{(k)}=X_{T}^{(k)}$ for all $t>T$), and given Condition \ref{cond:damping}, then we have a contraction mapping with $\lambda_t^{(k')}$ converging to a unique fixed point as $k'\rightarrow \infty$, namely, $\lambda_t^{(k+1)} $ for all $t\leq T$. 

\subsection{Convergence to a Fixed-Point}
The two-step iteration described in Section \ref{sec:2stepPointScheme} looks for a fixed point of $\lambda^{(k')}$ for a given $X^{(k+1)}$. In this section we present a theorem stating that the pair $(X^{(k)},\lambda^{(k)})$ given by \eqref{eq:iterationScheme} converges to a fixed point, thus confirming that the forward-backward system of \eqref{eq:Xequation} and \eqref{eq:lambdaEquation} has a unique solution. We prove these results for the finite-time version of \eqref{eq:Xequation} and \eqref{eq:lambdaEquation} with terminal condition $\lambda_{T+1}^{(k)} \equiv 0$ for all $k$. 

\begin{theorem}
    \label{thm:convThm}
    Consider the finite-time problem with terminal condition $\lambda_{T+1}^{(k)} \equiv 0$ for all $k$. If we assume Condition \ref{cond:C}, Condition \ref{cond:S_LB}, Condition \ref{cond:q_bound} and Condition \ref{cond:damping}, then the iterations of \eqref{eq:iterationScheme} will converge to a unique fixed point.
\end{theorem}
\begin{proof} (see Appendix).
\end{proof}
\begin{remark}
    Accuracy of the approximation of the infinite-time problem by a finite-time problem can be proven with optimality bounds and a squeeze lemma as $T\rightarrow \infty$.
\end{remark}

The main idea in the proof of Theorem \ref{thm:convThm} is to show a contraction in $\mathbb E\|\lambda_t^{(k+1)}-\lambda_t^{(k)}\|$, thus confirming the existence of a unique fixed point. The initial step for setting up the proof is to write the following forward equation for the iteration difference,
\begin{align}
    \label{eq:fullyForwardForm}
    &\lambda_t^{(k+1)}-\lambda_t^{(k)}=\delta(I+\widetilde P_t\Psi_t^{-1})^{-1}\mathbb E_t\left[ \lambda_{t+1}^{(k+1)}-\lambda_{t+1}^{(k)}\right] -\gamma (I+\widetilde P_t\Psi_t^{-1})^{-1} P_t\left(\sum_{t'=1}^{t-1}\frac{1}{\epsilon q(\vec S_{t'},P_t)}\Psi_{t'}^{-1}(\lambda_{t'}^{(k)}-\lambda_{t'}^{(k-1)})\right) \
\end{align}
which is derived from \eqref{eq:iterationScheme} by differencing between iteration $k+1$ and $k$. The following quantity is important for convergence,
\begin{align}
    \label{eq:Delta_and_Omega}
    \Omega &= \frac{ \goodchi\|h\|_\infty}{\underbar s}\ .
\end{align}
 The quantity in \eqref{eq:Delta_and_Omega} is important because from \eqref{eq:fullyForwardForm} we obtain the following inequality,
\begin{align*}
    \sup_{t\leq T} &\mathbb E\left\|\lambda_t^{(k+1)}-\lambda_t^{(k)} \right\| \leq \frac{\Omega }{1-\Delta(\epsilon)}\sup_{t\leq T}\sum_{t'=1}^{t-1}\mathbb E\left\|\lambda_{t'}^{(k)}-\lambda_{t'}^{(k-1)}\right\|\ ,
\end{align*}
An equivalent form of this inequality appears in Theorem \ref{thm:convThm} and is used to show that $\mathbb E\left\|\lambda_t^{(k+1)}-\lambda_t^{(k)}\right\|\rightarrow 0$ as $k\rightarrow \infty$ for $t\leq T$, thus proving the existence of a unique fixed point of \eqref{eq:iterationScheme} for finite $T$.

\subsection{Neural Network Policy Approximation}
The machine learning approach to solve \eqref{eq:iterationScheme} is to implement RL using neural networks (NNs). It amounts to estimating the limiting fixed point $\lambda^*$ from \eqref{eq:iterationScheme} as
\[\lambda_t^*\approx \lambda(t,X_{t-1},\vec S_t,P_t;\theta)\ \]
where $\lambda(\cdot,\cdot,\cdot,\cdot;\theta)$ is a policy approximation function with a feed-forward  NN, and $\theta$ denotes the NN parameters to be estimated. The NN takes into account time $t$ so that the solution can be adapted to the time remaining until terminal time $T$. An example of an NN that can accommodate time dependence is the Deep BSDE architecture in \cite{han2018deepBSDE}.

Given an initial estimate $\theta^{(0)}$, we proceed to iteratively look for $\theta$ that is close to a fixed point. The following iterative estimation scheme is the basis for the algorithm we'll implement: 
\begin{align}
    \label{eq:implicitScheme}
    &\theta^{(k+1)}= \argmin_\theta~\sum_{t=1}^T\mathbb E\Big\|\lambda(t,X_{t-1}^{(k+1)},\vec S_t,P_t;\theta) - Y_t^{(k)}(\theta) \Big\|^2\hbox{ for } \\
    \nonumber
    &\hspace{1cm}\hbox{s.t.}\\
    \nonumber
    &X_t^{(k+1)}= X_{t-1}^{(k+1)} -\frac{1}{\epsilon q(\vec S_t,P_t)}\Psi_t^{-1}\lambda(t,X_{t-1}^{(k+1)}, \vec S_t,P_t;\theta^{(k)})\\
    \nonumber
    &Y_t^{(k)}(\theta)=\delta\left(I+\widetilde P_t\Psi_t^{-1}\right)^{-1}  \mathbb  E_t\Big[\lambda(t+1,X_t^{(k+1)},\vec S_{t+1},P_{t+1};\theta)\Big]\mathbbm{1}_{t<T}-\left(I+\widetilde P_t\Psi_t^{-1}\right)^{-1}\Bigg(\Psi_t\mu -\gamma P_tX_{t-1}^{(k+1)}\Bigg)\,
\end{align}
where the indicator $\mathbbm{1}_{t<T}$ is used to enforce the finite-time problem's terminal condition of $\lambda_{T+1}=0$. 
The following theorem proves the NN scheme in \eqref{eq:implicitScheme} will in fact be a decent approximation for the fixed point $\lambda_t$.

\begin{theorem}
    \label{thm:NN_fixed_point}
    Assume Condition \ref{cond:C}, Condition \ref{cond:S_LB} and Condition \ref{cond:q_bound}. Furthermore, assume the family $$(\lambda(t,X_{t-1}^{(k )},\vec S_t,P_t;\theta^{(k)}))_{k=1,2,\dots}$$ are continuous and uniformly integrable, i.e., for any $\eta>0$ there is compact set $\mathbf K\subset \mathbb R^d\times \mathbb R^d\times \mathbb R^{d\times d}$ on which we have
    \begin{align}
        \label{eq:UI}
        \sup_k\mathbb E\Big[\|\lambda(t,X_{t-1}^{(k    )},\vec S_t,P_t;\theta^{(k)})\| \mathbbm{1}_{\{(X_{t-1}^{(k)},\vec S_t,P_t)\notin \mathbf K\}}\Big]\leq \eta\ .
    \end{align}
    Assume also that $$\mathbb E\Big[\|\lambda^*(t,X_{t-1}^*,\vec S_t,P_t)\|\mathbbm{1}_{\{(X_{t-1}^*,\vec S_t,P_t)\notin \mathbf K\}}\Big]\leq \eta.$$ For each $k$ let $\varepsilon^{(k)}$ denote the error from the NN approximation,
    \begin{align*}
        \sup_{t\leq T} &\mathbb E\Big[\left\|\lambda(t,X_{t-1}^{(k)},\vec S_t,P_t;\theta^{(k)}) - Y_t^{(k-1)}(\theta^{(k)})\right\| \mathbbm{1}_{\{(X_{t-1}^{(k)},\vec S_t,P_t)\in \mathbf K\}}\Big]\leq  \varepsilon^{(k)}\ .
    \end{align*}
    Then, the error of the iteration scheme in \eqref{eq:implicitScheme} is
    \begin{align}
        \label{eq:NNbound}
        \sup_{t\leq T} &\mathbb E\left\|\lambda(t,X_{t-1}^{(k+1)},\vec S_t,P_t;\theta^{(k+1)})-\lambda_t^*\right\| =\mathcal O\left(\left(\sup_\ell\varepsilon^{(\ell)}+2\eta\right)\exp\Big(\frac{\Omega T}{1-\Delta(\epsilon)}\Big)\right)\ ,
    \end{align} 
for $k$ large, where $\lambda^*$ is the fixed point of \eqref{eq:iterationScheme}.
\end{theorem}

\begin{proof} (see Appendix).

\end{proof}

The bound in \eqref{eq:NNbound} is similar to the bounds derived in \cite{pham2018deep}, wherein the approximation error was computed to be the sum of three terms: a sampling error term, an NN parameter error, and a value function estimation error. An additional similarity is an exponential growth in the bounding constants as $T$ increases. In theory, the exponential growth in \eqref{eq:NNbound} is contained by continually increasing the hyperparameters so that $\sup_\ell\varepsilon^{(\ell)}+2\eta$ tends to zero, but in practice increasing hyperparameters requires a growing number of training samples, leading to an infeasibly long computation time. However, the smallness of $\epsilon$ can be exploited for a faster algorithm, which does not eliminate exponential growth in $T$, but we are at least able to reduce bounding constants with lowered computational cost.


\section{Small-$\epsilon$ Asymptotic Analysis \& Implementation}

Let's consider the parameterization with $\epsilon$ being small enough that effects of order $\epsilon^2$ can be dropped or grouped in with round-off error. In this setting, we can construct a solution in a power series form and then truncate terms order $\epsilon^2$ and higher. In other words, our expansion has a base term equal to the explicitly computable myopic solution (obtained for the case $\delta=0$), and an order-$\epsilon$ correction term. Computation of the order-$\epsilon$ correction is more involved, but computational costs and runtime are minimal.

\subsection{Expansion of $\lambda_t$}
\label{sec:coStateExpansion}
We write the formal expression for $\lambda_t$,
\begin{align*}
    \lambda_t &= \epsilon \widetilde \lambda_t \ ,
\end{align*}
which we insert into \eqref{eq:Xequation} and \eqref{eq:lambdaEquation} to obtain the following equations,
\begin{align}
    \label{eq:XtildeLambdaSystem}
	X_t &= X_{t-1}- \frac{1}{q(\vec S_t,P_t)}\Psi_t^{-1}\widetilde\lambda_t\\
	\label{eq:LambdatildeLambdaSystem}
	\widetilde\lambda_t& = \delta\mathbb E_t\widetilde\lambda_{t+1}-\frac{\gamma}{\epsilon }P_t\left(\frac{1}{\gamma}P_t^{-1}\Psi_t\mu - X_t\right)\\
	\nonumber
	& = \delta\mathbb E_t\widetilde\lambda_{t+1}-\frac{\gamma }{\epsilon  }P_t\left(\frac{1}{\gamma}P_t^{-1}\Psi_t\mu - X_{t-1}\right) -  \frac{\gamma }{\epsilon q(\vec S_t,P_t) }P_t\Psi_t^{-1}\widetilde\lambda_t\ .
\end{align}
Rearranging \eqref{eq:LambdatildeLambdaSystem} we have the following stabilized equation,
\begin{align}
    \label{eq:stabilized_tildeLambda}
	\widetilde\lambda_t
	 &=\left(\epsilon I+\frac{\gamma}{q(\vec S_t,P_t)} P_t\Psi_t^{-1}\right)^{-1} \left(\delta\epsilon \mathbb E_t\widetilde\lambda_{t+1}-\gamma P_t\left(\frac{1}{\gamma}P_t^{-1}\Psi_t\mu -  X_{t-1}\right)\right)\ ,
\end{align}
for which it is straightforward to check that the convergence proof of Theorem \ref{thm:convThm} still applies.
We write the following formal expansion,
\[\widetilde\lambda_t = \widetilde\lambda_t^{[0]}+\epsilon \widetilde\lambda_t^{[1]}+\epsilon^2 \widetilde\lambda_t^{[2]}+\dots\ ,\]
which we insert into \eqref{eq:stabilized_tildeLambda} to obtain the following recursive expressions for the expansion's terms,
\begin{align}
    \label{eq:expansionTerms}
    \widetilde\lambda_t^{[0]}&= -\left(\epsilon I+\frac{\gamma}{q(\vec S_t,P_t)} P_t\Psi_t^{-1}\right)^{-1}\left(\Psi_t\mu - \gamma P_t X_{t-1}\right)\\
    \nonumber
    \widetilde\lambda_t^{[i]}&=
    \delta \left(\epsilon I+\frac{\gamma}{q(\vec S_t,P_t)} P_t\Psi_t^{-1}\right)^{-1}\mathbb E_t\widetilde\lambda_{t+1}^{[i-1]}
\end{align}
for $i=1,2,3,\dots$. This expansion is in powers of $\epsilon$ and can be truncated for a good approximation of the solution to \eqref{eq:XtildeLambdaSystem}. 

\begin{remark}
When we approximate $\widetilde\lambda_t$ with lower-order terms $\widetilde\lambda_t^{[0]}$ and $\widetilde\lambda_t^{[1]}$, we can simplify the expressions in \eqref{eq:expansionTerms} so that they do not depend on $\epsilon$
\begin{align}
    \label{eq:expansionNaive}
    \widetilde\lambda_t^{[0]}&= 
    -\left(\frac{\gamma}{q(\vec S_t,P_t)} P_t\Psi_t^{-1}\right)^{-1}\left(\Psi_t\mu - \gamma P_t X_{t-1}\right)\\
    \nonumber
    \widetilde\lambda_t^{[1]}&=
    \delta \left(\frac{\gamma}{q(\vec S_t,P_t)} P_t\Psi_t^{-1}\right)^{-1}\mathbb E_t\widetilde\lambda_{t+1}^{[0]} +\left(\frac{\gamma}{q(\vec S_t,P_t)} P_t\Psi_t^{-1}\right)^{-2}\left(\Psi_t\mu - \gamma P_t X_{t-1}\right)\ ,
\end{align}
so that $\widetilde\lambda_t =\widetilde\lambda_t^{[0]}+\epsilon\widetilde\lambda_t^{[1]}+\mathcal O(\epsilon^2)$.The expansion in \eqref{eq:expansionNaive}  is different from \eqref{eq:expansionTerms} because it has reduced the base term to the naive policy that goes straight to the aim portfolio, thus leaving it to the correction terms to compensate for transaction costs.
\end{remark}
Inserting the $\widetilde\lambda_t$ expansion into \eqref{eq:XtildeLambdaSystem}, we suspect the following lower-order expansion for the state process has a Big-Oh error as follows,
\begin{align}
    \label{eq:XtildeEpsilonProcess}
    X_t &= X_{t-1}-\frac{1}{q(\vec S_t,P_t)}\Psi_t^{-1}\left(\widetilde\lambda_t^{[0]}+\epsilon \widetilde\lambda_t^{[1]}\right)+ \mathcal O(\epsilon^2)\ .
\end{align}
These order $\mathcal O(\epsilon^2)$ errors are in fact true, based on the following proposition:

\begin{proposition}
    \label{prop:expansionAccuracy}
    Assume Condition \ref{cond:C}, Condition \ref{cond:S_LB} and Condition \ref{cond:q_bound}. The order-$\epsilon$ approximation $\widetilde\lambda_t^{[0]}+\epsilon \widetilde\lambda_t^{[1]}$ has an error, which is
    \begin{align*}
        \sup_{t\leq T}&\mathbb E\left\|\widetilde\lambda_t^{[0]}+\epsilon\widetilde\lambda_t^{[1]}-\widetilde\lambda_t^*\right\|= 
\mathcal O\left(\epsilon\Delta(\epsilon)\exp\Big(\frac{\Omega T}{1-\Delta(\epsilon)}\Big)\mathbb E\sup_{t\leq T}\|\widetilde\lambda_t^{[1]}\|\right)\ ,
    \end{align*} 
    where $\widetilde\lambda^*$ denotes the solution to \eqref{eq:XtildeLambdaSystem} and \eqref{eq:LambdatildeLambdaSystem}.
\end{proposition}
\begin{proof} (see Appendix.)
\end{proof}

\subsection{Neural Network Algorithm}
\label{sec:NNalgo}

Let function $\varphi:\mathbb R^d\times \mathbb R^d\times \mathbb R^{d\times d}\rightarrow \mathbb R^n$ be from a class of NN functionals with sigmoidal activation. The term $\widetilde\lambda_t^{[1]}$ is approximated as $\delta\left(\epsilon I+\frac{\gamma}{q(\vec S_t,P_t)} P_t\Psi_t^{-1}\right)^{-1}\varphi(X_{t-1},\vec S_t,P_t;\theta)$, where the optimal parameter $\theta$ is found from
\[\min_\theta\sum_{t=1}^T\mathbb E\left\|\varphi(X_{t-1},\vec S_t,P_t;\theta) - \mathbb E_t[\widetilde\lambda_{t+1}^{[0]}]\mathbbm{1}_{t< T}\right\|^2\ .\]
Note that now we are considering a NN architecture that remains constant through time, which is a considerable simplification from the NN architecture used in the proof of Theorem \ref{thm:NN_fixed_point}. The tradeoff in making this simplification is faster computation time. Algorithm \ref{alg2} gives the implementation of the scheme in \eqref{eq:implicitScheme} using the lower-order expansion in \eqref{eq:XtildeEpsilonProcess} with this NN approximation of $\delta\mathbb E_t\widetilde \lambda_{t+1}^{[0]}$. Fig.~\ref{fig:flow} shows the flow of Alg.~\ref{alg2} at each moment $t$. We empirically observed that after training, our NN policy can work in real-time (specifically, each NN inference takes around 1.5 ms on a commodity laptop). 

In our analysis of Algorithm \ref{alg2}'s portfolio, we will compare to the purely myopic strategy, 
\begin{align}
    \label{eq:myopic}
    X_t^{\hbox{\tiny myopic}}&=X_{t-1}^{\hbox{\tiny myopic}}+ \frac{1}{\epsilon q(\vec S_t,P_t)}\Psi_t^{-1}\left(I+\widetilde P_t\Psi_t^{-1}\right)^{-1}\left(\mu - \gamma P_t X_{t-1}^{\hbox{\tiny myopic}}\right)\ .
\end{align}
The myopic strategy shares the same objective as the RL strategy but for  $\delta=0$. Note that because our objective function contains a quadratic transaction cost term, the typical buy-and-hold portfolios will create a large transaction cost at the beginning, leading to a negative total wealth return. Therefore, we do not compare our method with the  typical buy-and-hold portfolios. 

\begin{algorithm}[t!] 
 \caption{Small-$\epsilon$ Neural Network Fixed-Point Algorithm for MGARCH, with Learning Rate $\alpha\in(0,1]$}  
 \label{alg2} 
 \begin{algorithmic} 
     \STATE Initialize: $\theta^{(1)} \sim \mathcal{N}(0, 0.01)$.
     \FOR{k = 1 to MAX\_ITER}
        \STATE Initialize: $P_0$, $ \vec S_0$ and $X_0$.
        \STATE Set: $ S_0 = \hbox{diag}(\vec S_0) $.
     	\FOR{t = 1 to T}
     	\STATE \#\#\# Update the MGARCH state:
		    \STATE $Z_t = R_t-\mu$
			\STATE $\Sigma_t = CC^\top+A\Sigma_{t-1}A^\top + BZ_tZ_t^\top B^\top$ 
			\STATE $ S_t = S_{t-1}(1+R_t)$
            \STATE $\vec S_t = (S_t^1, S_t^2, ..., S_t^n)^T$
			\STATE $\Psi_t = \hbox{diag}(\vec S_t)$
            \STATE $P_t = \Psi_t \Sigma_t \Psi_t$
		\STATE \#\#\# Update the portfolio:
		\STATE $\widetilde\lambda_t^{[0]} = -\left(\epsilon I+ \frac{\gamma}{ q(\vec S_t,P_{t})}P_{t}\Psi_t^{-1}\right)^{-1}\Big(\Psi_t\mu-\gamma P_tX_{t-1}\Big)$
		\STATE $\widetilde\lambda_t^{[1]} = \delta\left(\epsilon I+ \frac{\gamma}{ q(\vec S_t,P_{t})}P_{t}\Psi_t^{-1}\right)^{-1}\varphi(X_{t-1},\vec S_t,P_t;\theta^{(k)})$
		\STATE $X_t = X_{t-1} - \frac{1}{ q(\vec S_t,P_t)}\Psi_t^{-1}\Big(\widetilde\lambda_t^{[0]}+\epsilon \widetilde\lambda_t^{[1]}\Big)$

	\ENDFOR
         \STATE $\text{loss}(\theta) = \frac{\sum_{t=1}^T\left\|\varphi(X_{t-1},\vec S_t,P_t;\theta) -  \widetilde\lambda_{t+1}^{[0]}\mathbbm{1}_{t<T}\right\|^2}{T-1}$
         \STATE $\theta^* =\argmin_\theta \text{loss}(\theta)$
         \STATE $\theta^{(k+1)} = \alpha \theta^* + (1-\alpha)\theta^{(k)}$
     \ENDFOR
 \end{algorithmic}
 \end{algorithm}

\begin{figure}
    \centering
    \includegraphics[scale=0.35]{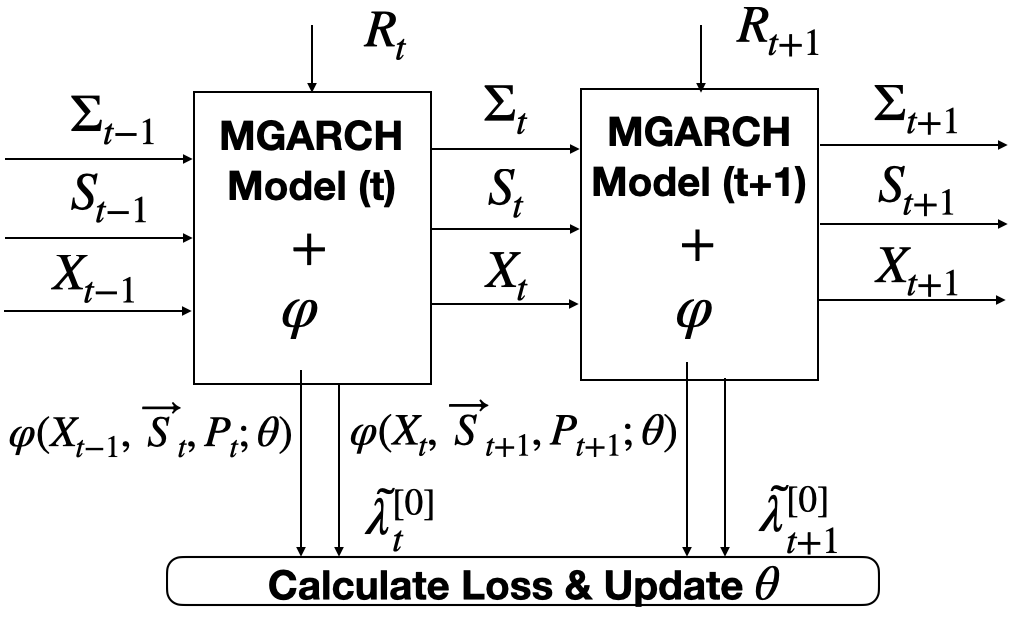}
    \caption{Flow of Alg.~\ref{alg2}}
    \label{fig:flow}
\end{figure}

\section{Sector ETFs: an 11-Dimensional Example}

We verify the small-$\epsilon$ asymptotic analysis by allocating 11 sector ETFs, including iShares U.S. Real Estate ETF (IYR), Materials Select Sector SPDR Fund (XLB), Energy Select Sector SPDR Fund (XLE), Financial Select Sector SPDR Fund (XLF), Industrial Select Sector SPDR Fund (XLI), Technology Select Sector SPDR Fund (XLK), Consumer Staples Select Sector SPDR Fund (XLP), Utilities Select Sector SPDR Fund (XLU), Health Care Select Sector SPDR Fund (XLV), Consumer Discretionary Select Sector SPDR Fund (XLY), and Vanguard Communication Services ETF (VOX). The chosen ETFs are a good representation of a cross-section of the 
U.S. stock market returns. However, our method is applicable to any set of stocks. Training of Algorithm~\ref{alg2} was implemented on synthetic data sampled from an MGARCH model that was estimated on historical data of the 11 ETFs. We will show the implementation results of Algorithm~\ref{alg2} on both simulated testing data, and on historical out-of-sample data that was not used to estimate the MGARCH model. Overall, our algorithm practically performs well. Additionally,  it is not necessary that the rigorous conditions of Sec.~\ref{sec2} have to hold, as we do not enforce Condition \ref{cond:S_LB}'s boundedness of $h$. 

\subsection{Setup}

\subsubsection{Dataset}
We downloaded the adjusted closing price of the 11-sector ETFs in the year 2010 to the year 2019 from Yahoo Finance (\cite{yf2020}). The data was split into 10 folds for training, with each fold having 5 years of historical data. The starting and ending dates for the folds are shown in Table~\ref{table:fold}. 

\begin{table}[t!]
\caption{Details of Each Fold}
\begin{center}
\begin{tabular}{ccc|ccc}
\toprule
\textbf{Fold} & \textbf{From} & \textbf{To} & \textbf{Fold} & \textbf{From} & \textbf{To}\\
\hline
1 & Jan. 2010 & Oct. 2014  &
2 & Jul. 2010 & Apr. 2015  \\
3 & Jan. 2011 & Oct. 2015  &
4 & Jul. 2011 & Apr. 2016  \\
5 & Feb. 2012 & Nov. 2016  &
6 & Aug. 2012 & May 2017  \\
7 & Mar. 2013 & Nov. 2017  &
8 & Sep. 2013 & Jun. 2018  \\
9 & Feb. 2014 & Dec. 2018  &
10 & Aug. 2014 & Jun. 2019  \\
\bottomrule
\end{tabular}
\end{center} 
\label{table:fold}
\end{table}

\subsubsection{Estimating Parameters}

Given a sequence of historical price $\{S_t^i\}$ for the $i^{th}$ ETF, the return rate $R_t^i$ at each time $t$ can be found by:
$$R_t^i = \frac{S_{t+1}^i- S_t^i}{S_t^i}.$$
Defining $R_t = [R_t^1, R_t^2, ..., R_t^{11}]^\top $ and $\Bar R$ the mean of $R_t$ over $t$, the initial covariance matrix $\Sigma_0$ is calculated as:
$$\Sigma_0 = \frac{1}{L} \sum_{t=1}^L (R_t - \Bar R)  (R_t - \Bar R)^\top,$$ where $L$ is the length of the sequential historical data.

$A$, $B$, and $C$ in (\ref{eq:volMatrix}) were estimated using  Broyden Fletcher Goldfarb Shanno (BFGS) algorithm (\cite{BFGS13}). Specifically, we first burned in the initial covariance matrix $\Sigma_0$ into  (\ref{eq:volMatrix}). We then used $\sum_{t\geq 0}||\Sigma_{t+1} - \Sigma_{t}||^2$ with matrix norm as the loss function and BFGS as the minimizer to find $A$, $B$, $C$, and $\Sigma$ to minimize the loss function.  The following equilibrium state is achieved when the loss function is the smallest:
\begin{align}
     \Sigma= CC^\top +A \Sigma A^\top + B \Sigma B^\top.
    \label{equilibrium}
\end{align}
For the estimation of $\mu$, we used the eigen-portfolio approach in \cite{boyle2014}.
The parameters (i.e., $A$, $B$, $C$, $\Sigma_0$, and $\mu$) were estimated for each fold. Fig.~\ref{price} shows an example of  the historical ETF prices and the simulated ETF prices using the estimated parameters.

\begin{figure}[ht!]
        \centering
        \includegraphics[scale=0.3]{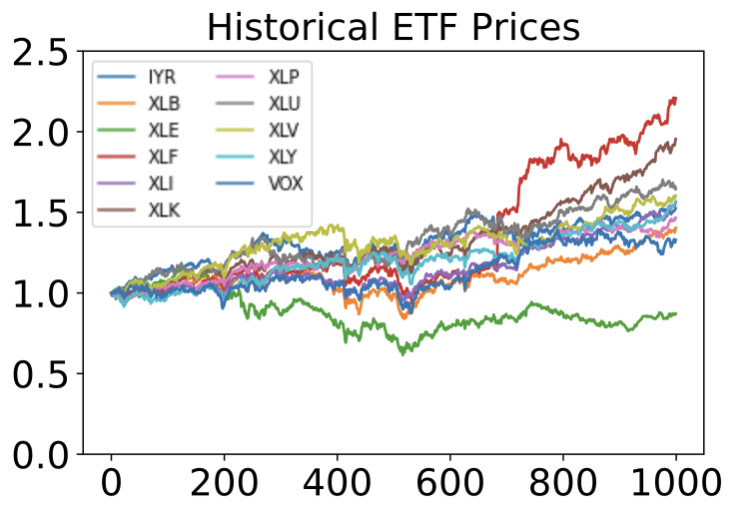}
        \includegraphics[scale=0.3]{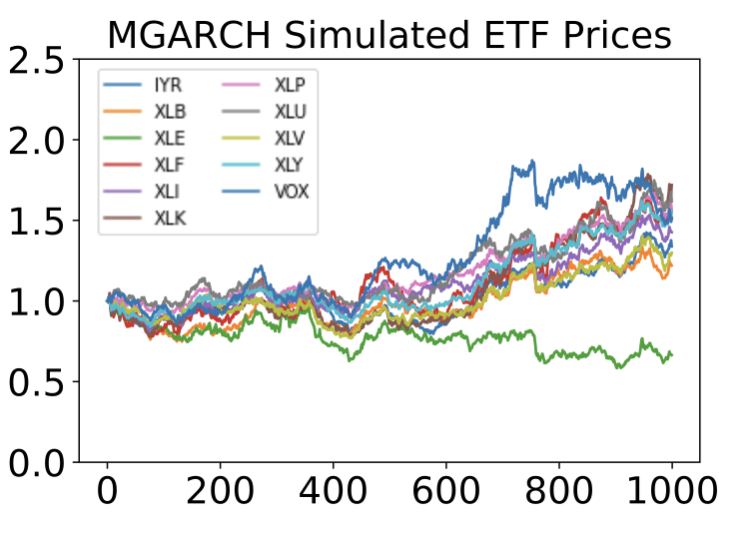}
        \caption{ X-axis shows the day. Y-axis shows the ETF prices. }
        \label{price}
\end{figure}

\subsubsection{Architecture of Neural Network}

The utilized neural network (NN) is composed of fully connected layers.   The input contains the portfolio $X_{t-1}$ that has 11 elements, plus the covariance matrix of the dollar-returns $P_t = \Psi_t\Sigma_t\Psi_t\ $ whose dimension is $11\times11$ and also the expected value of returns $\mu^\top \Psi_t$ which is also 11-dimensional. Therefore, the total dimension is 143. The hidden layer size was determined by considering both the training time and the NN performance. When using more complex NN architecture, we observed that there was no obvious improvement in performance while the training time increased considerably. On the other hand, when using even simpler NN architectures, we observed that the deep RL algorithm suffered from under-fitting problems. Therefore, we utilized four hidden layers, each of which contains 400 neurons. The output of the NN corresponds to $\varphi(\cdot,\cdot, \cdot;\theta)$ in Algorithm~\ref{alg2}, which is 11-dimensional. The activation function is Tanh.  The architecture of the utilized neural network (NN) is shown in Table~\ref{table:architecture} in the appendix shows the detail of the architecture.
The programming language is Python 3. Tensorflow and Keras were utilized as the deep-learning library. Keras is built on top of Tensorflow. 

\subsection{Simulation on Synthetic Data}
\label{Syn}
At moment $t+1$, a noise vector $Z_{t+1}$ was generated from a Gaussian distribution $\mathcal{N}(0, \Sigma_t)$, where $\Sigma_t$ is the covariance matrix at the previous moment, the return rate $R_{t+1}$ was determined by (\ref{eq:returns}), and the covariance matrix $\Sigma_{t+1}$ was found by \eqref{eq:volMatrix}. For each fold, we generated one sequence of synthetic data to train the NN and the other 200 sequences to test the performance of the NN. We show how the total wealth grows with time for RL and myopic strategies. We also show how the fund manager should invest the money into the stock market with time using RL and myopic strategies, respectively.
Specifically, wealth $W_t$ at time $t$ was calculated by
$$W_t = W_{t-1}+X_{t-1}^\top(\vec S_t - \vec S_{t-1})-\frac{\epsilon q(\vec S_t, P_t)}{2}a_t^\top \Psi_t a_t\ ,$$
with $W_0 = 100$. We ran the trained NN and the myopic strategy (as given by \eqref{eq:myopic}) on the test data for each fold and calculated the average. The training epoch was 100 for each fold. The optimizer was Adam. The loss function was mean-squared-error loss (MSE). In this case, we took $\epsilon=0.003$, for which we were able empirically to verify Condition \ref{cond:damping}. The results are shown in Fig.~\ref{plots}.
\begin{figure}[ht!]
        \centering
        \includegraphics[scale = 0.4]{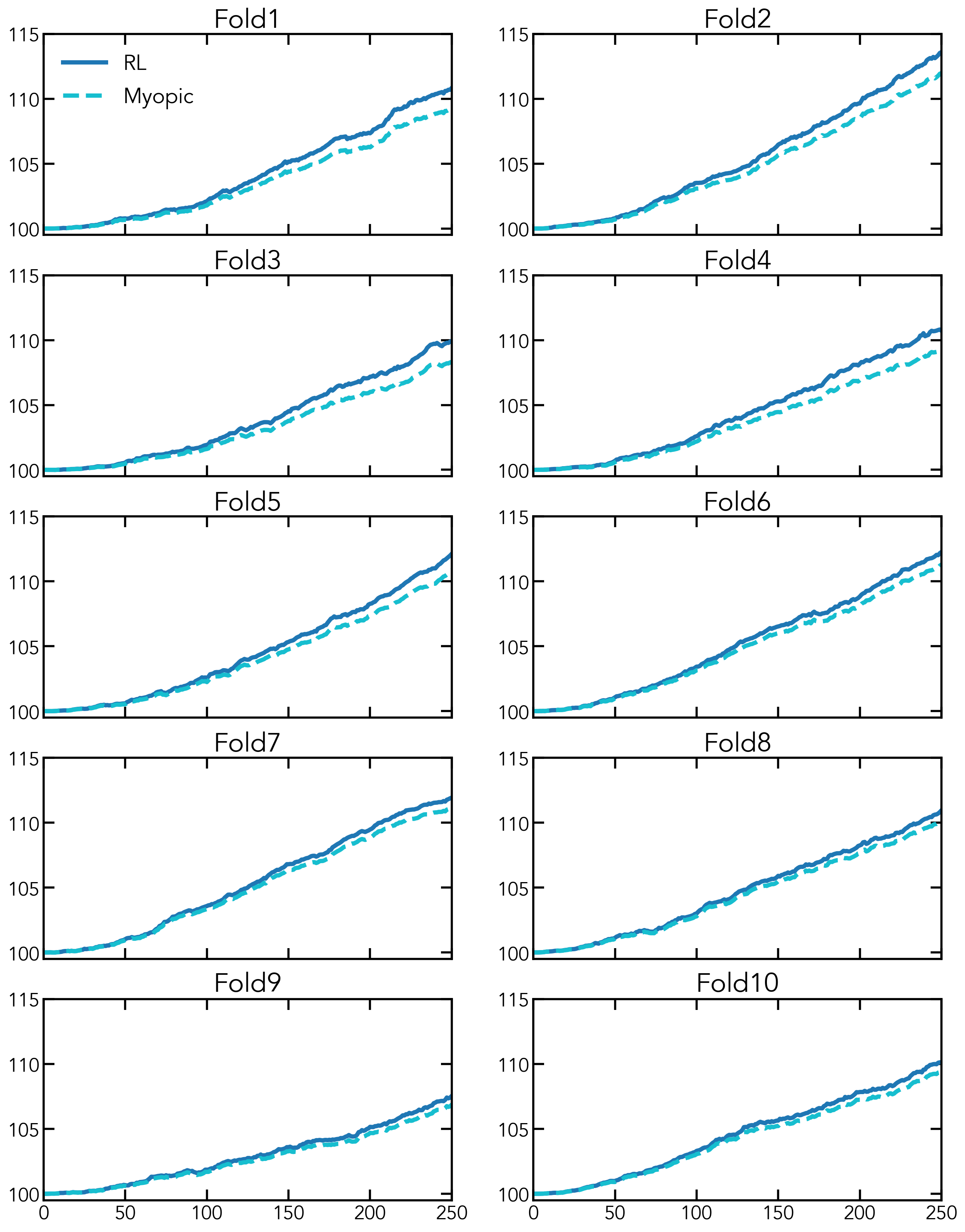}
        \caption{Average wealth (200 trajectories) of the RL strategy (solid) and the myopic strategy (dot). The X-axis shows the day. Y-axis shows the total wealth. }
        \label{plots}
\end{figure}

Table~\ref{table:ratio} shows the value added of the two methods in each fold on the synthetic data. 
The RL strategy of Algorithm \ref{alg2} consistently outperforms the myopic strategy given by \eqref{eq:myopic} by showing an annual increase of 1.8\% in total value. The results are scalable to any $W_0$. Therefore, our RL approach will show significant outperformance in the absolute value of total assets added when the given capital (i.e., $W_0$) is large.
Note in the table that we are emphasizing the \$-value of returns, as this is a better measure of fund performance as per the reasoning of \cite{berk2015measuring}.

\begin{table}
\caption{Final Total Asset Value Added on Synthetic Data with $W_0=\$100$.}
\begin{center}
\begin{tabular}{|c|cc|cc|c|}
\hline
& \multicolumn{2}{c}{$W_T$ (\$)} & \multicolumn{2}{|c|}{$W_T/W_0-1$ (\%)} & $W_T(RL)/$ \\
Fold & RL & myopic & RL  & myopic & $W_T(myopic)$ \\
\hline
1 & 107.57 & 105.45 & 7.57 & 5.45 & 1.020 \\
2 & 107.62 & 105.68 & 7.62 & 5.68 & 1.018 \\ 
3 & 104.46 & 103.10 & 4.46 & 3.10 &1.013 \\
4 & 105.68 & 103.88 & 5.68 & 3.88 & 1.017\\ 
5 & 106.29 & 104.61 & 6.29 & 4.61 & 1.016\\
6 & 111.87 & 109.18 & 11.87 & 9.18 &1.024 \\ 
7 & 109.55 & 107.55 & 9.55 & 7.55 & 1.018\\
8 & 110.43 & 108.17 & 10.43 & 8.17 & 1.020\\ 
9 & 106.52 & 104.86 & 6.52 & 4.86 & 1.015\\
10 & 109.13 & 107.16 & 9.13 & 7.16 & 1.018\\
\hline
Ave. & 107.91 & 105.96 & 7.91 & 5.96 & 1.018\\
\hline
\end{tabular}
\end{center} 
\label{table:ratio}
\end{table}

The amount of money invested in the stock market is calculated by
$$ I_t = I_{t-1} + a_{t-1}^\top \vec S_{t-1}\ ,$$ where $I_0 = 0$ because all money is in cash at the beginning. Note that since the ETF prices vary with time, it is likely that at the end of the day, the total money invested in the stock market may exceed $W_0$.
Fig.~\ref{investment} shows the evolution of $I_t$ for both RL and myopic strategies. The two strategies show different investment speeds. The RL is investing faster than the myopic strategy. Fig.~\ref{portfolio} shows the evolution of the RL and myopic portfolio allocations with time. In the beginning, the portfolio allocations are at $0$ for both RL and myopic. As time passes, the portfolios increase, but the slope becomes smaller. Eventually, the portfolios will be attracted to a stationary state. It is worth mentioning that the investment portfolio charts are not for comparison purposes. Instead, they are used to  illustrate how our RL approach and the myopic approach are applied and to ensure that the two approaches have reasonable behaviors.

\begin{figure}[ht!]
        \centering
        \includegraphics[scale=0.4]{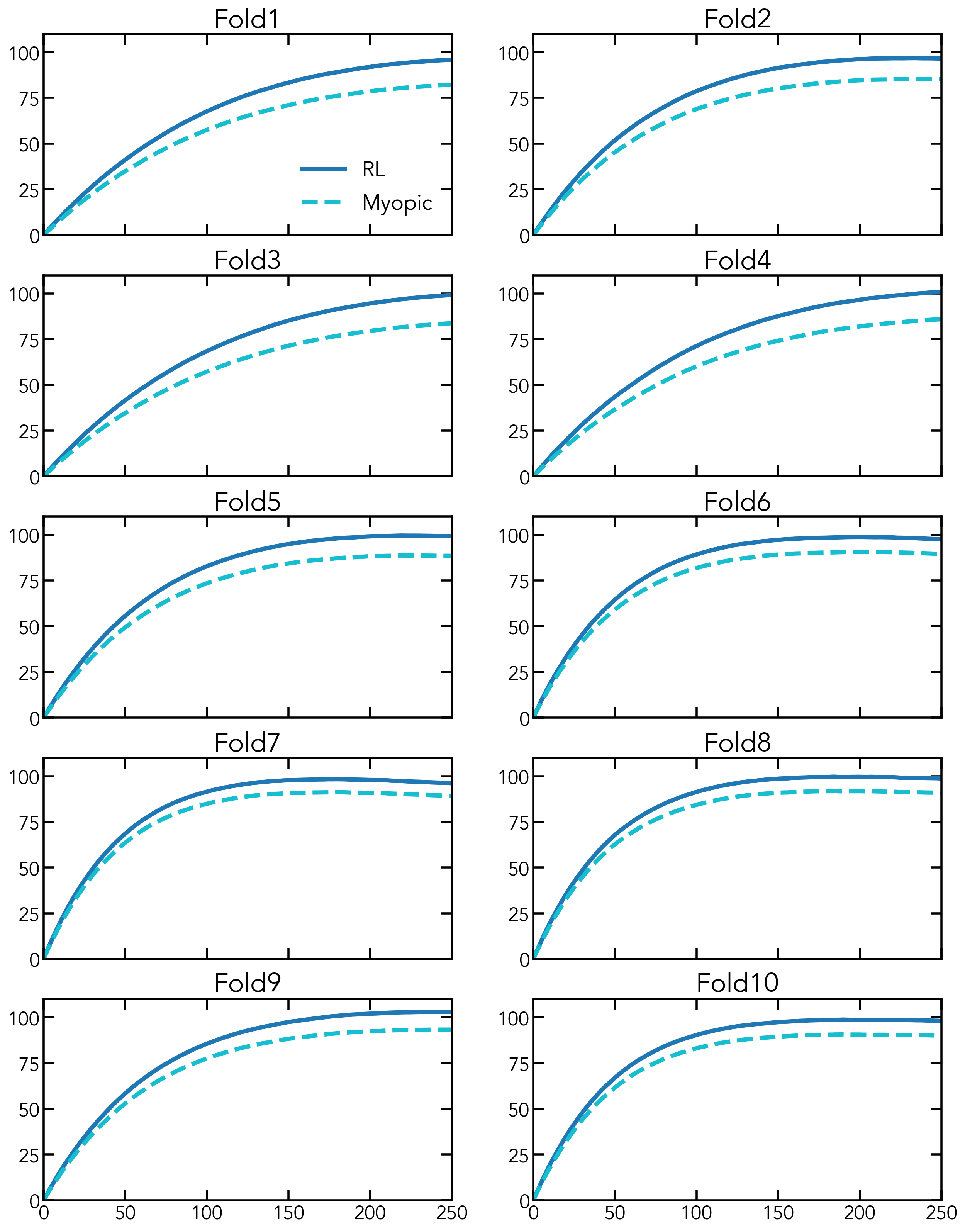}
        \caption{Money invested in the stock market at a different time for the RL strategy (solid) and the myopic strategy (dot). The X-axis shows the day. Y-axis shows the money invested in the stock market.}
        \label{investment}
\end{figure}

\begin{figure}[ht!]
        \centering
        \includegraphics[scale=0.4]{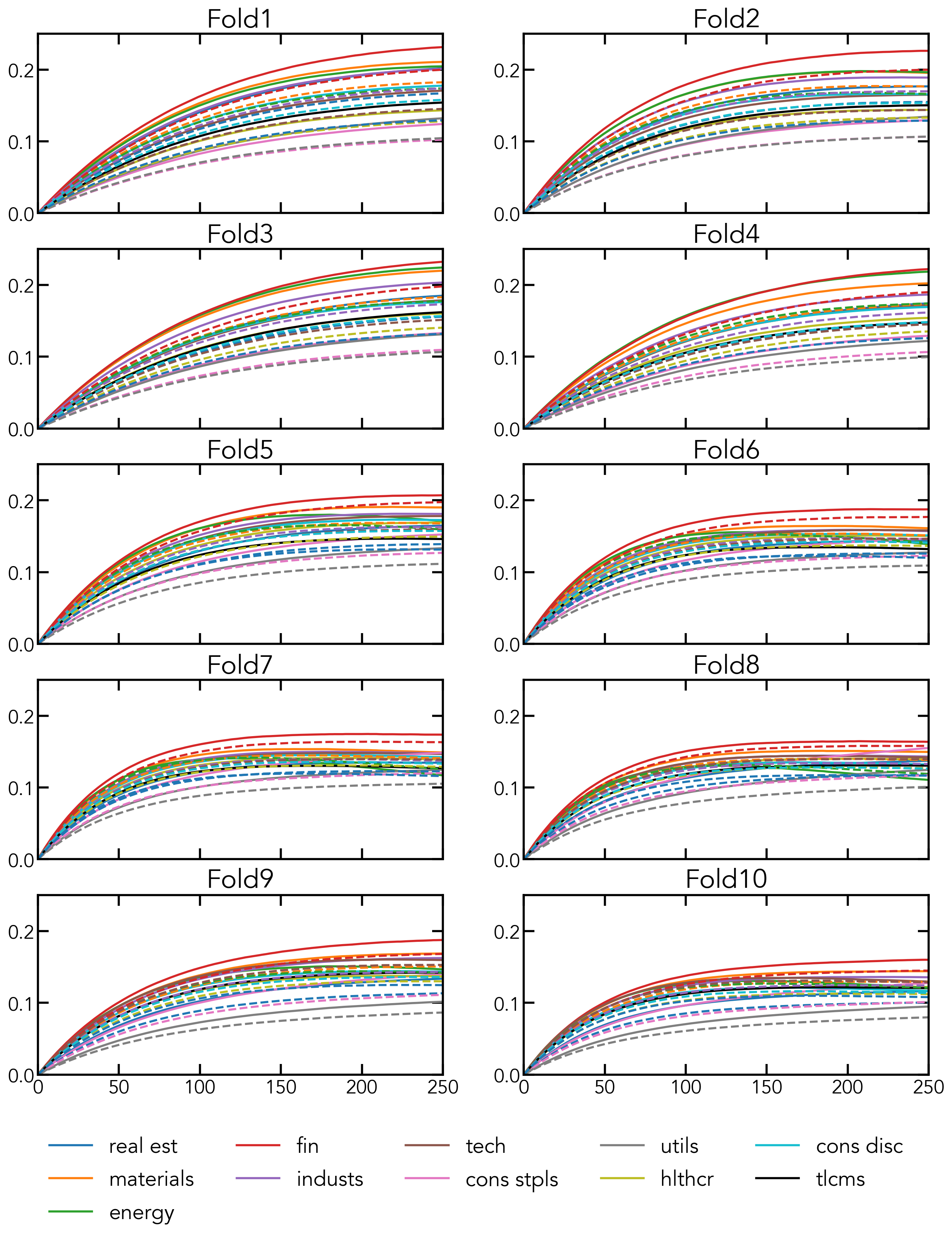}
        \caption{Portfolios of the RL strategy (solid) and the myopic strategy (dot) at different times. The X-axis shows the day. Y-axis shows the portfolio.}
        \label{portfolio}
\end{figure}

\subsection{Simulation on Historical Data}

For historical out-of-sample data, $Z_t$ was found by $$Z_t = R_t - \mu\,$$ where $R_t$ is the return rate of the real historical data at each time. For each fold, the NN was trained with the data of that fold and tested out-of-sample on the historical data of the following six months. The training epoch was 100. The optimizer was Adam. And the loss function was MSE. In this case, we took $\epsilon=0.01$, and again we were able empirically to verify Condition \ref{cond:damping}. The result is shown in Fig.~\ref{historicalwealth}. For most situations, the RL outperforms the myopic strategy, but not as significantly as it did on synthetic data. This is because the real historical market data contains considerable uncertainty that might cause the diminished performance of the RL strategy in testing. Table~\ref{table:historicalratio} shows annualized total asset value added by using the two approaches on historical market data. Among the cases, the RL strategy shows more added value  than the myopic approach. The RL method is able to consistently outperform the myopic by about 11bps. This over-performance is not tremendous but certainly makes evident that RL can improve out-of-sample performance for this problem. The results are scalable: if the fund manager has 100 million dollars, using the RL strategy can help her gain an extra \$120,000 compared to the gain realized from the myopic strategy.

\begin{table}
\caption{Annualized Total Asset Value Added on Historical Data with $W_0=\$100$.}
\begin{center}
\begin{tabular}{|c|cc|cc|c|}
\hline
& \multicolumn{2}{c}{$W_T$ (\$)} & \multicolumn{2}{|c|}{$W_T/W_0-1$ (\%)} & $W_T(RL)/$ \\
Fold & RL & myopic & RL  & myopic & $W_T(myopic)$ \\
\hline
1+2 & 99.58 & 99.57 & -0.4148 & -0.4292 & 1.0001 \\
3+4 & 101.26 & 101.08 & 1.2669 & 1.0877 & 1.0017 \\
5+6 & 103.12 & 103.01 & 3.1297 & 3.0163 & 1.0011 \\
7+8 & 97.24 & 97.35 & -2.7593 & -2.6473 & 0.9988 \\
9+10 & 107.09 & 106.71 & 7.0979 & 6.7143 & 1.0035 \\
\hline
Ave. & 101.66 & 101.54 & 1.6640 & 1.5483 & 1.0011\\
\hline
\end{tabular}
\end{center} 
\label{table:historicalratio}
\end{table}

Fig.~\ref{historicalinvest} shows how the money is invested with time.  Fig.~\ref{historicalportfolio} shows how the portfolios change. It bears repeating that the practical purpose of this paper’s optimization is to optimally move a large amount of new capital into the market, i.e., bring this new capital into the fund. By showing the cumulative amount of money invested we can see how much of the investment goal has been accomplished. Figures \ref{investment} and \ref{historicalinvest} illustrate how the two strategies behave as they are moving the new capital into stocks. Overall, the portfolios show similar behavior to what we observed in the synthetic data case, namely, that the RL strategy seeks to move the new capital into the ETFs faster than the myopic.

Note that $\epsilon$, in this case, is larger than $\epsilon$ in Sec.~\ref{Syn}, this is because we want to highlight the difference between RL and myopic strategies. If $\epsilon=0.003$,  the difference between the two strategies is less easy to highlight. The reason for the weaker significance when testing out-of-sample is because of the following. First, Fig.~\ref{plots} shows persistent over-performance by RL because it is an average of 200 trajectories. It is possible that if we had abundantly more out-of-sample data, we could see averages playing out, in which case we might see a stronger out-of-sample performance by RL. Second, there may be an out-of-sample model risk, i.e., that the model changes or is misspecified in the out-of-sample test. In this case, the training is aimed at learning an optimum that becomes sub-optimal when applied in the out-of-sample test. Ultimately, we cannot say much more than this, but the different folds of data that we have shown do give us some sense for out-of-sample variation for both RL and the myopic portfolios. 

\newcommand\x{6cm}
\newcommand\y{3.5cm}
\begin{figure*}[ht!]
        \centering
        \includegraphics[width = \x, height = \y]{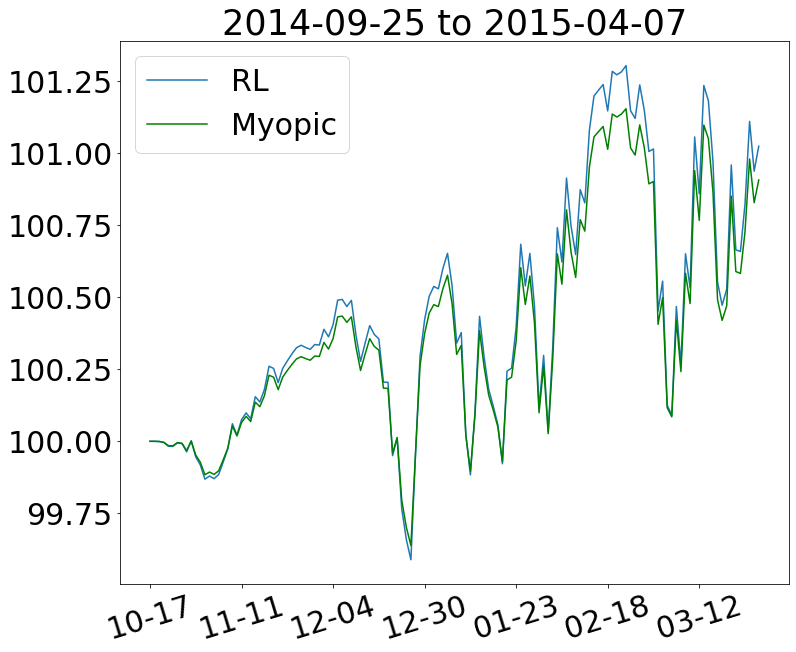}
        \includegraphics[width = \x, height = \y]{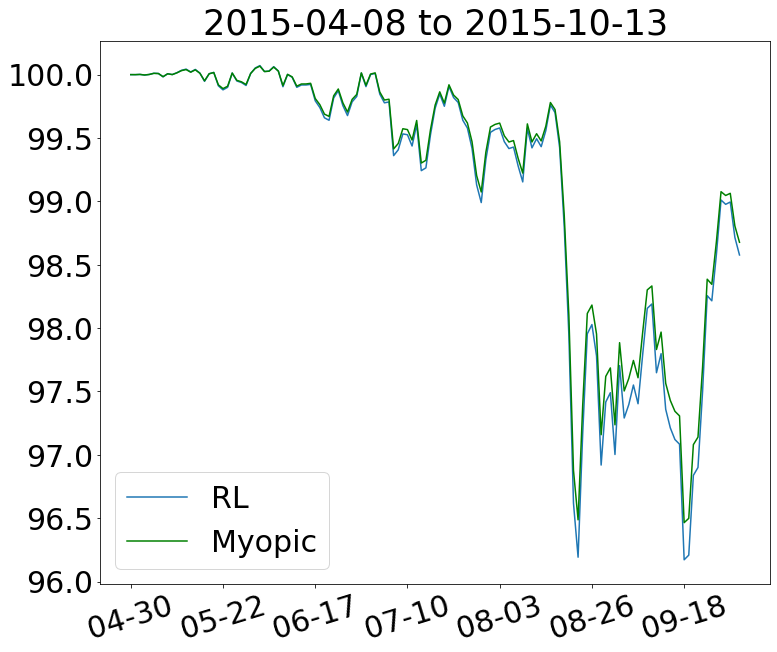}
        \includegraphics[width = \x, height = \y]{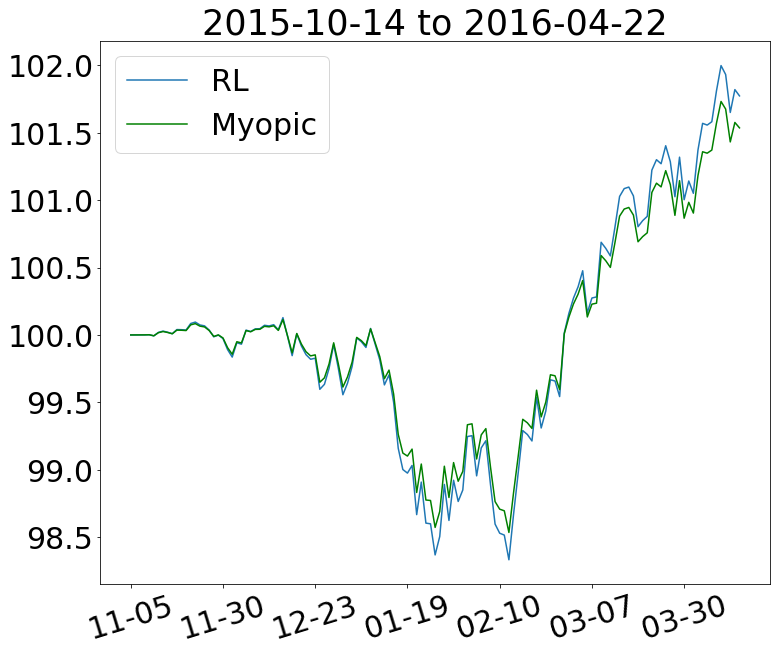}
        \includegraphics[width = \x, height = \y]{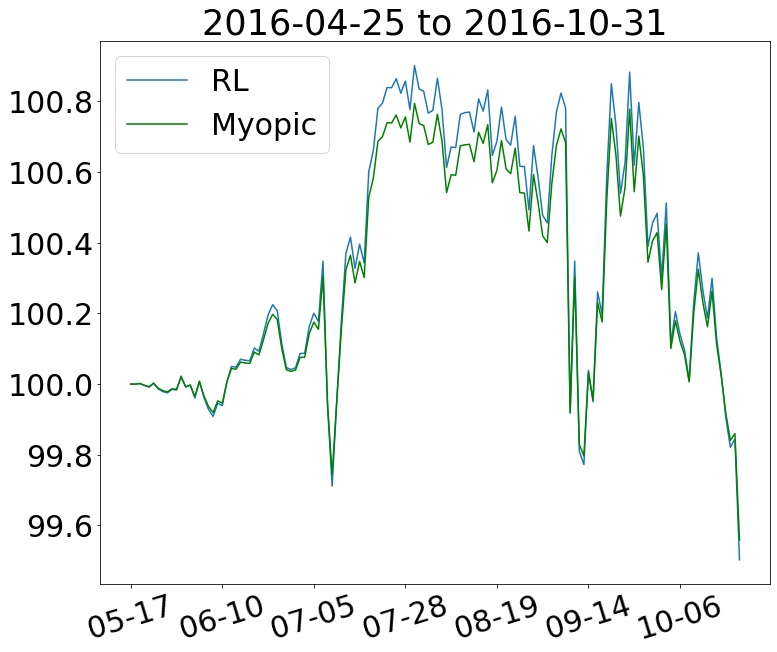}
        \includegraphics[width = \x, height = \y]{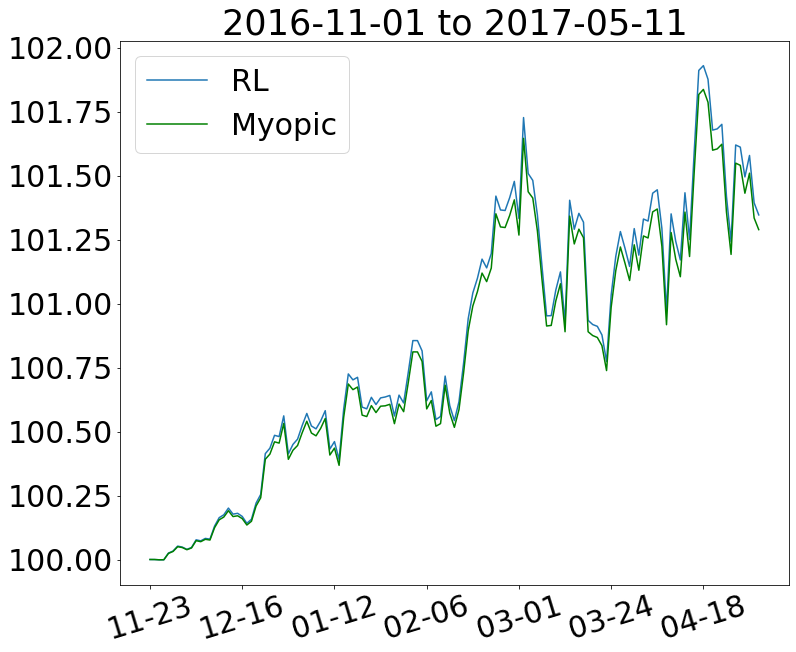}
        \includegraphics[width = \x, height = \y]{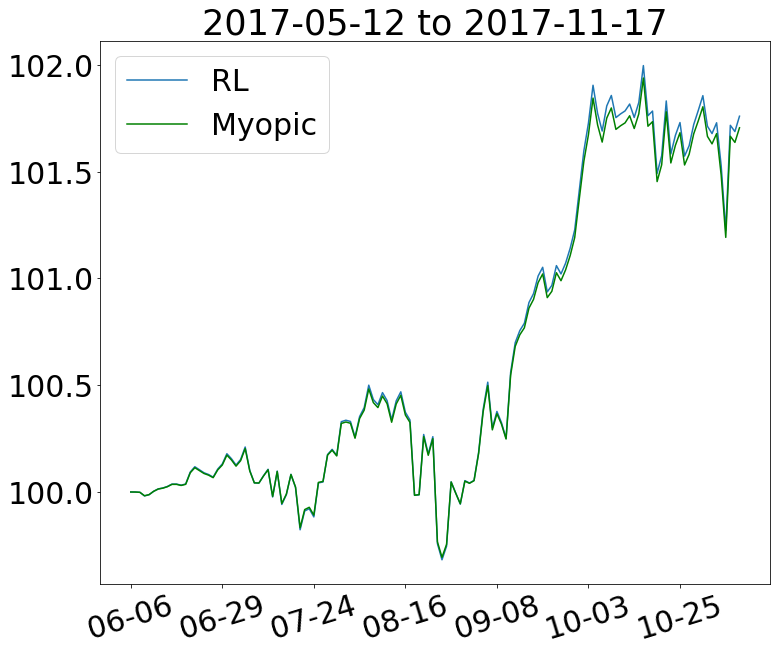}
        \includegraphics[width = \x, height = \y]{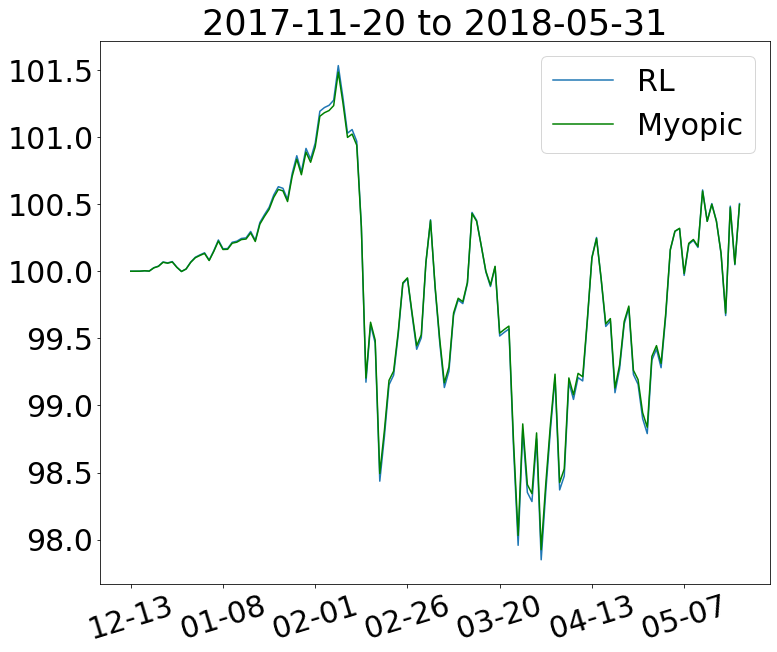}
        \includegraphics[width = \x, height = \y]{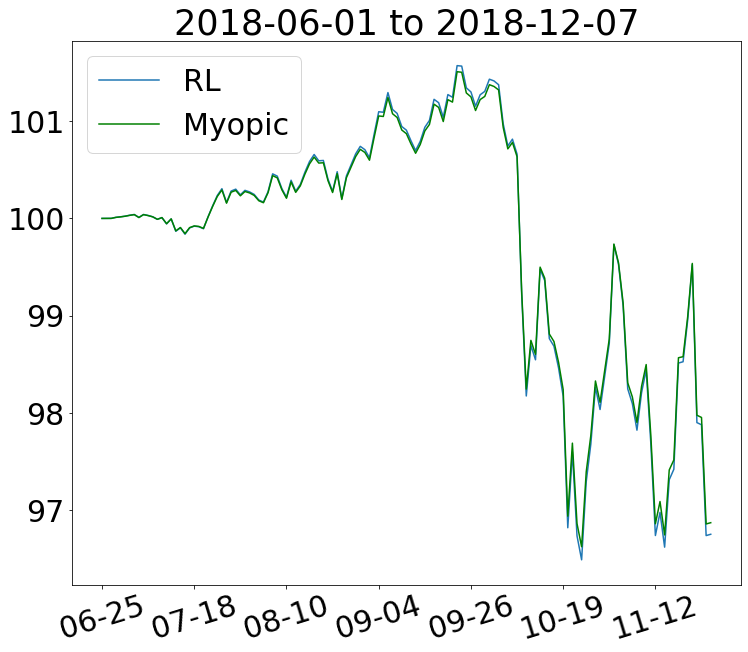}
        \includegraphics[width = \x, height = \y]{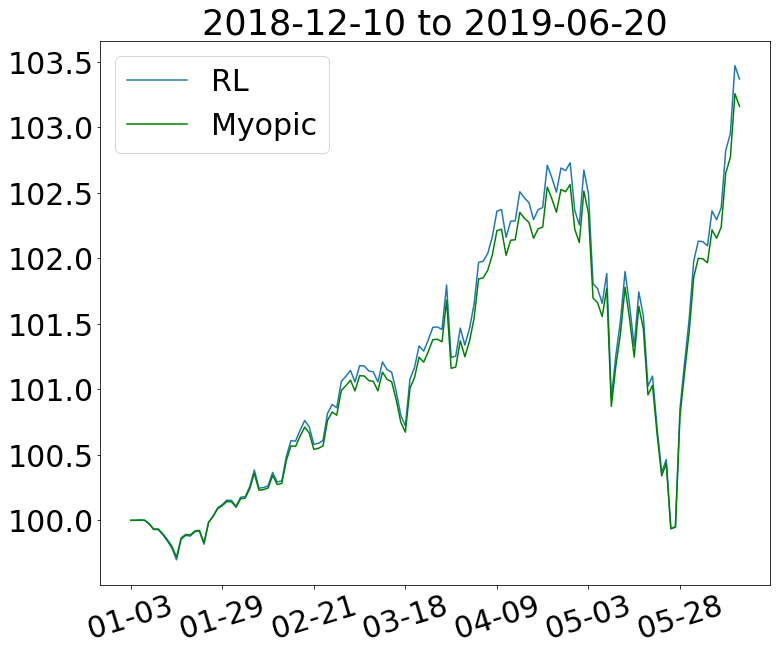}
        \includegraphics[width = \x, height = \y]{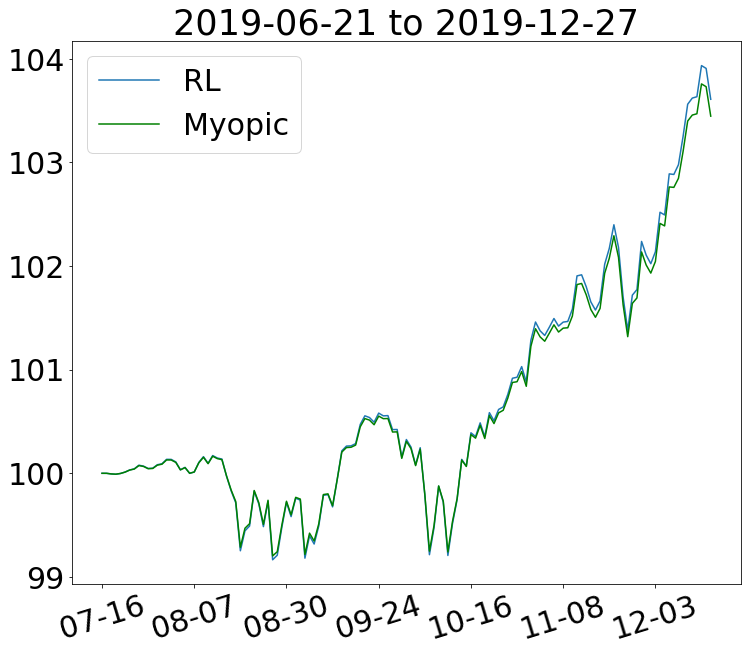}
        \caption{Total wealth of the RL strategy (solid) and the myopic strategy (dot) at different time.   X-axis shows the date (month - day).  Y-axis shows the wealth.}
        \label{historicalwealth}
\end{figure*}

\newcommand\xx{6cm}
\newcommand\yy{3.5cm}
\begin{figure*}[ht!]
        \centering
        \includegraphics[width = \xx, height = \yy]{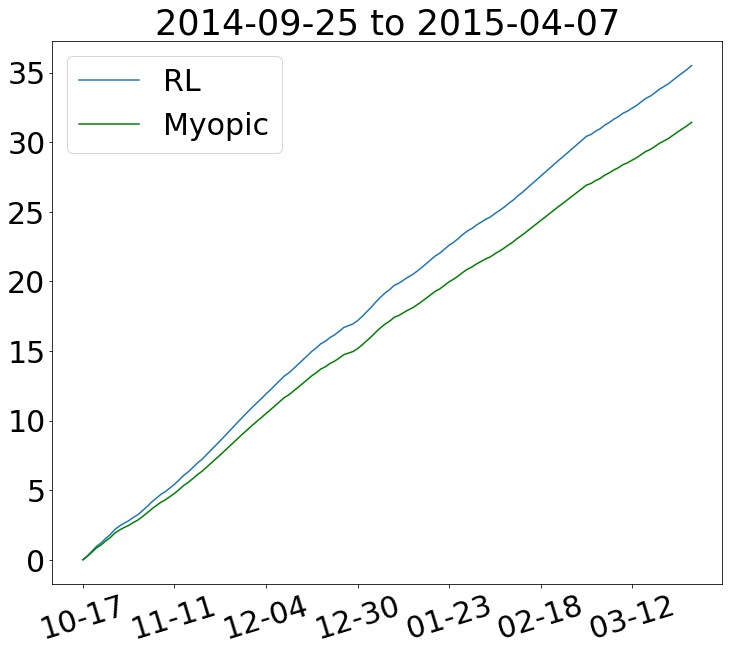}
        \includegraphics[width = \xx, height = \yy]{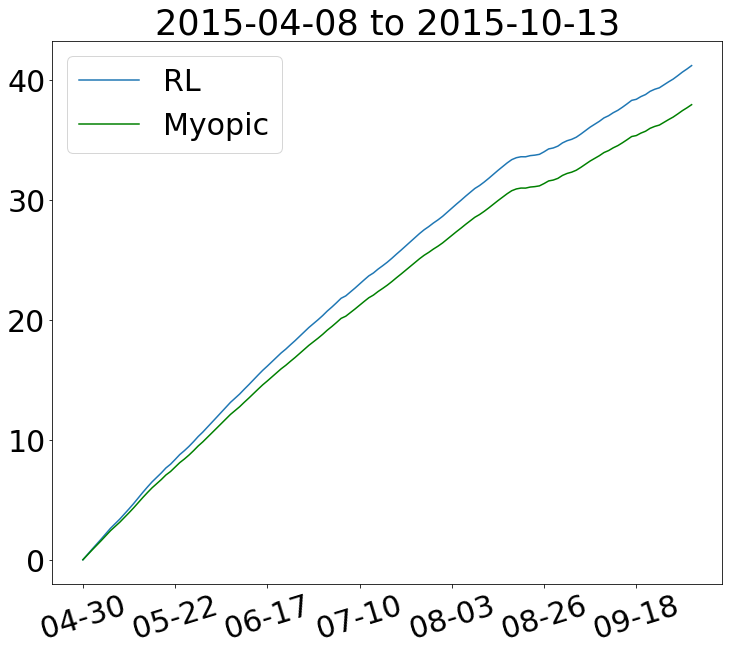}
        \includegraphics[width = \xx, height = \yy]{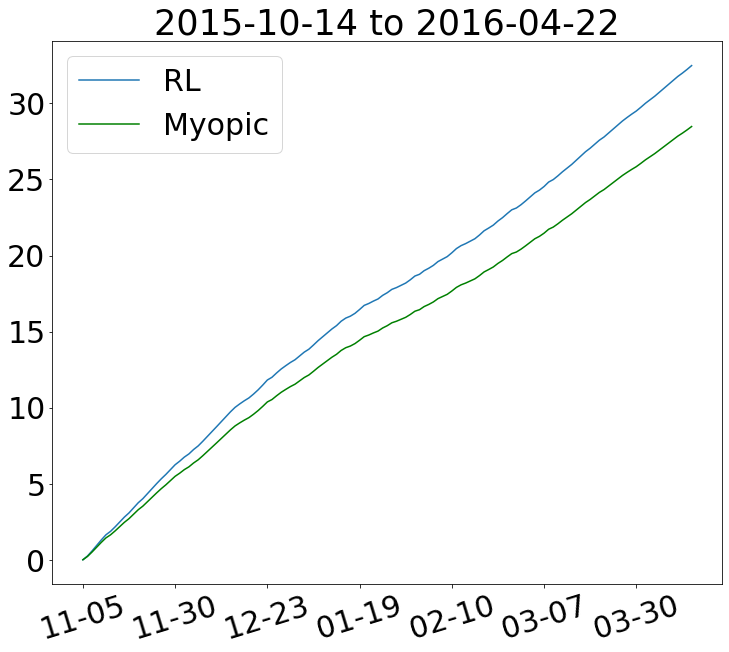}
        \includegraphics[width = \xx, height = \yy]{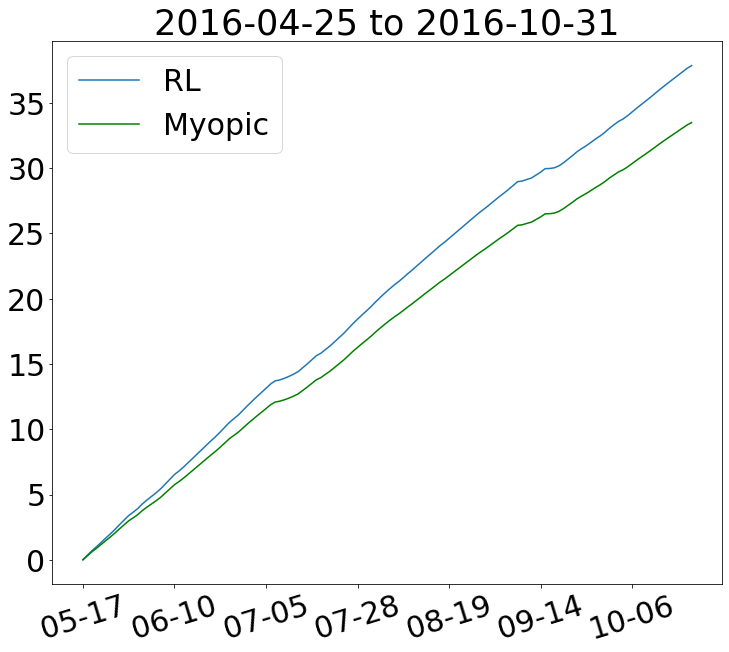}
        \includegraphics[width = \xx, height = \yy]{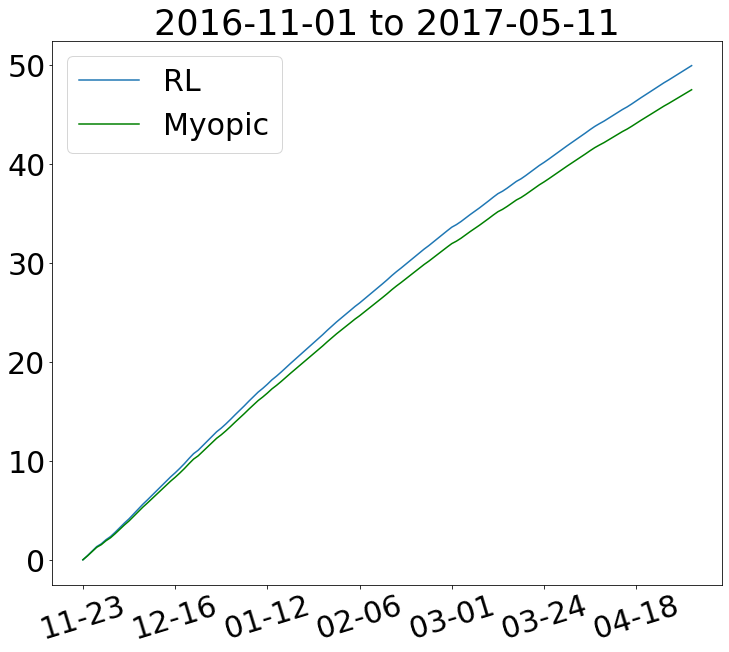}
        \includegraphics[width = \xx, height = \yy]{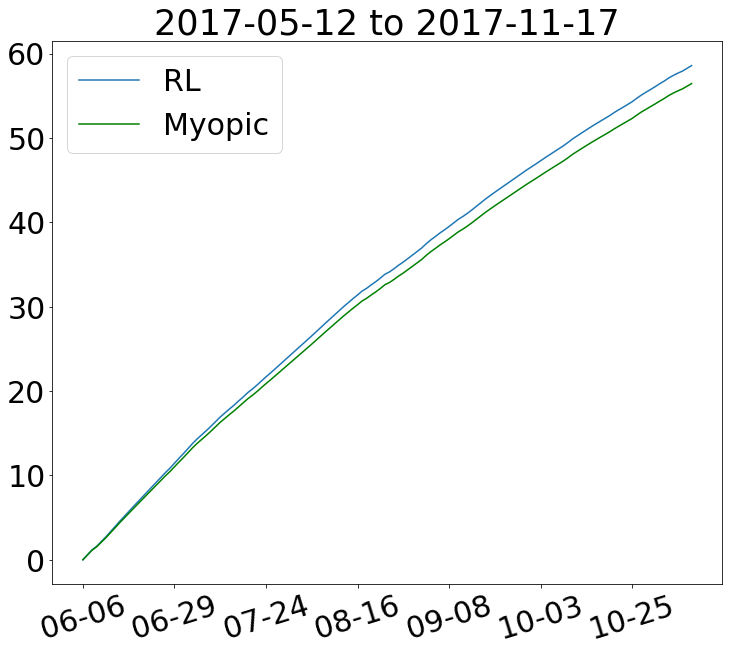}
        \includegraphics[width = \xx, height = \yy]{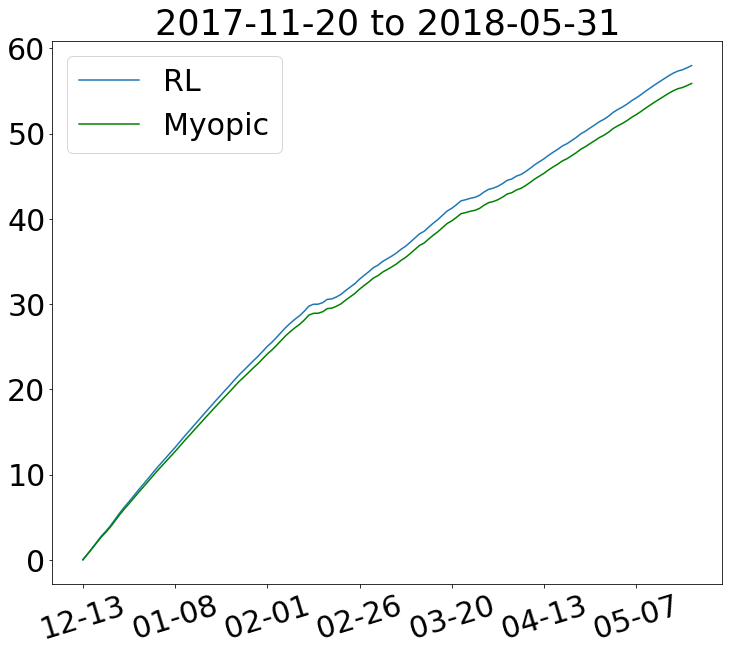}
        \includegraphics[width = \xx, height = \yy]{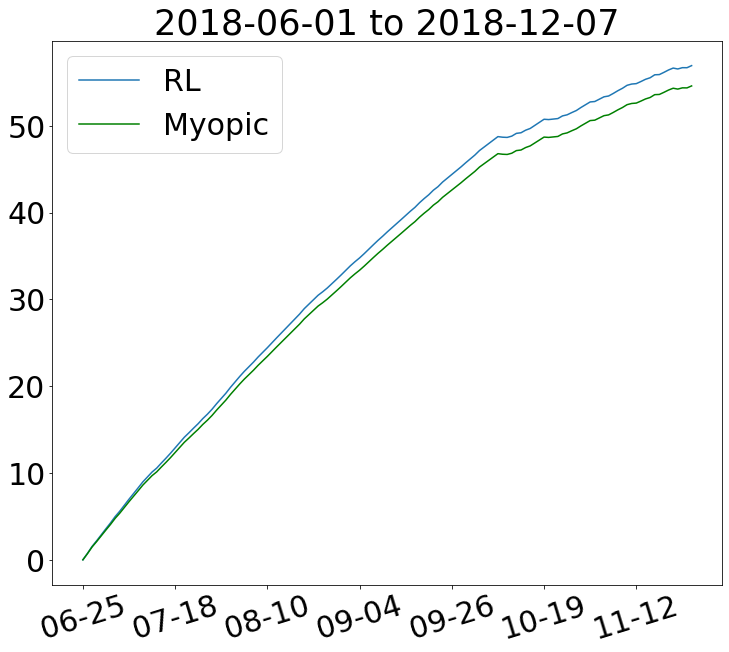}
        \includegraphics[width = \xx, height = \yy]{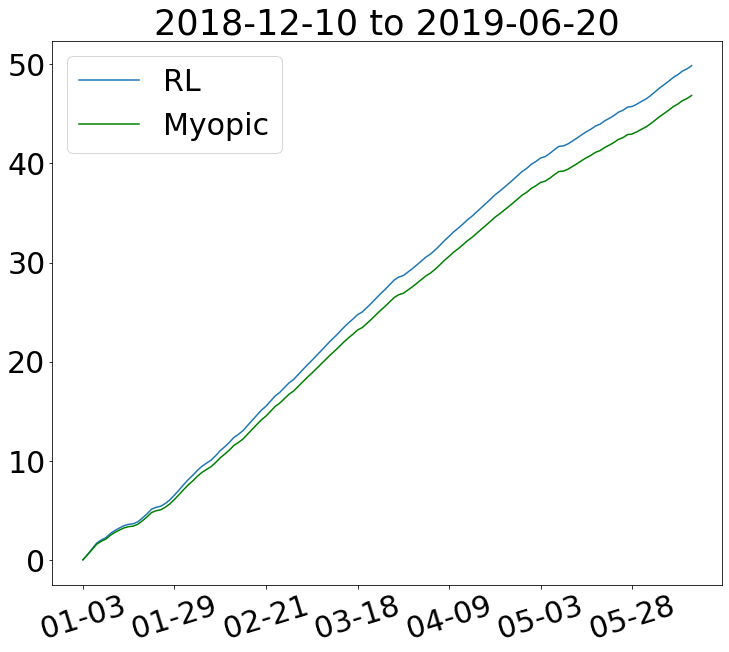}
        \includegraphics[width = \xx, height = \yy]{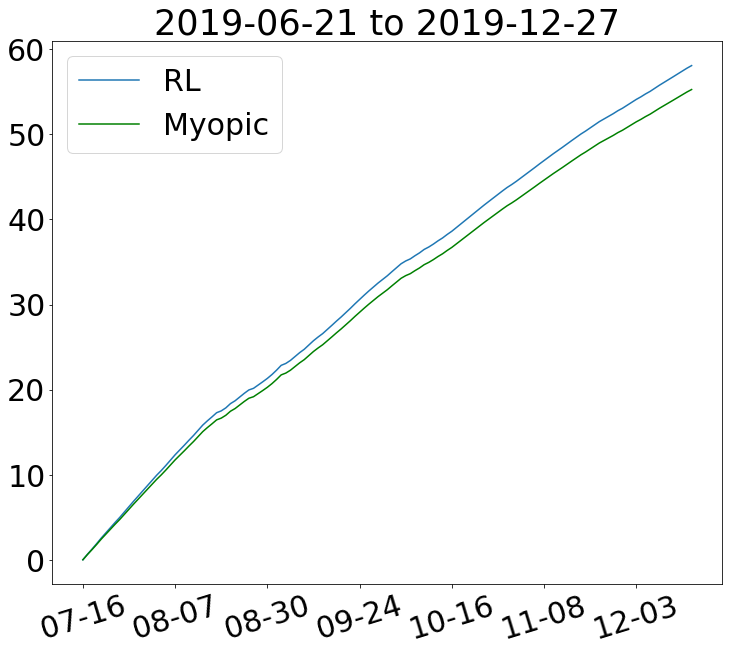}
        \caption{Money invested in the stock market at different times for the RL strategy (solid) and the myopic strategy (dot).   The X-axis shows the date (month - day).  Y-axis shows the money invested in the stock market.}
        \label{historicalinvest}
\end{figure*}

\begin{figure*}[ht!]
        \centering
        \includegraphics[width = \x, height = \y]{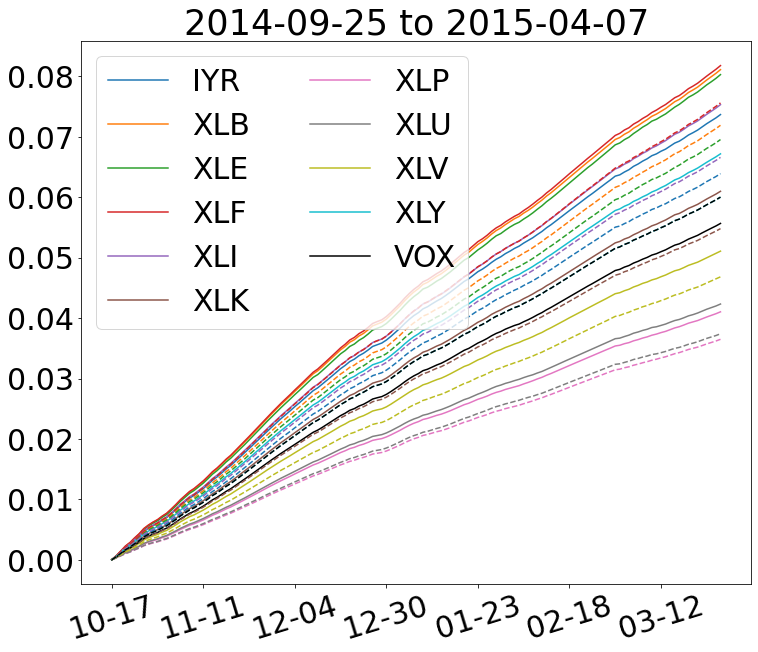}
        \includegraphics[width = \x, height = \y]{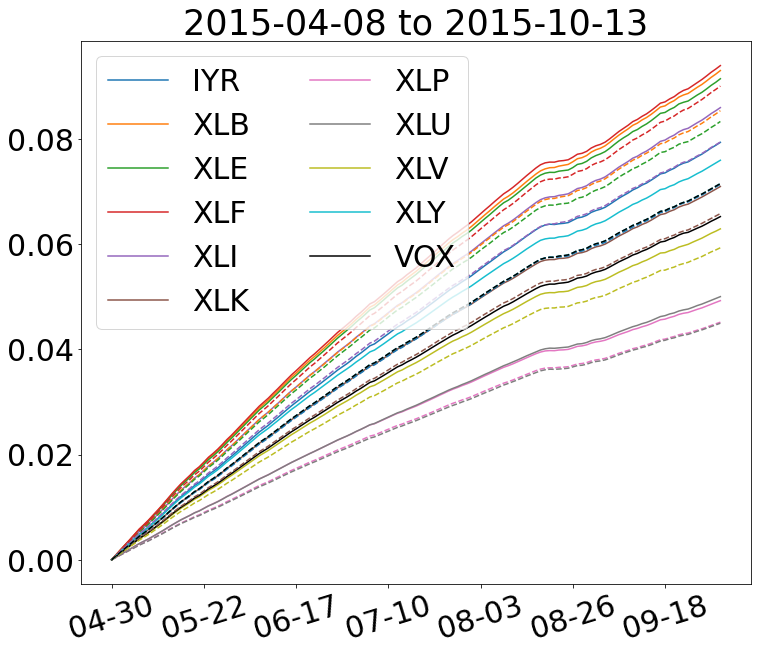}
        \includegraphics[width = \x, height = \y]{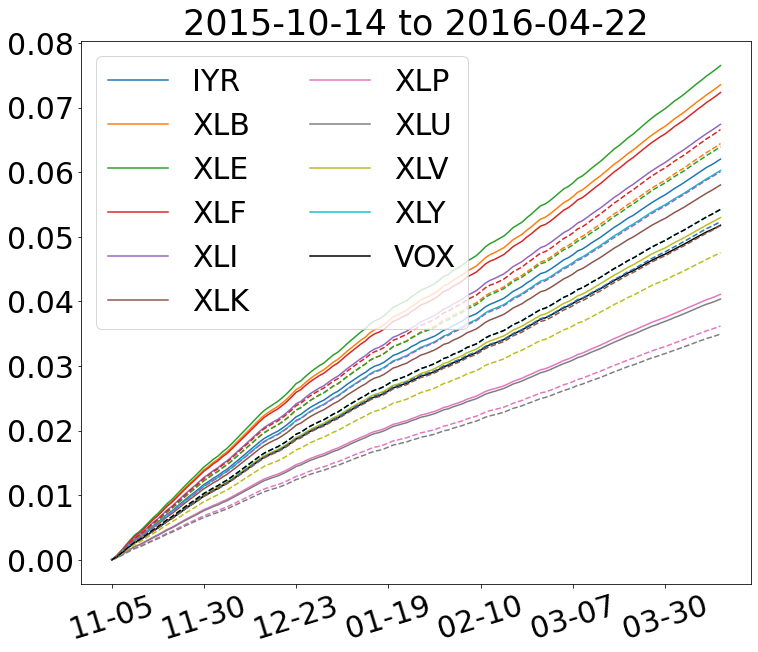}
        \includegraphics[width = \x, height = \y]{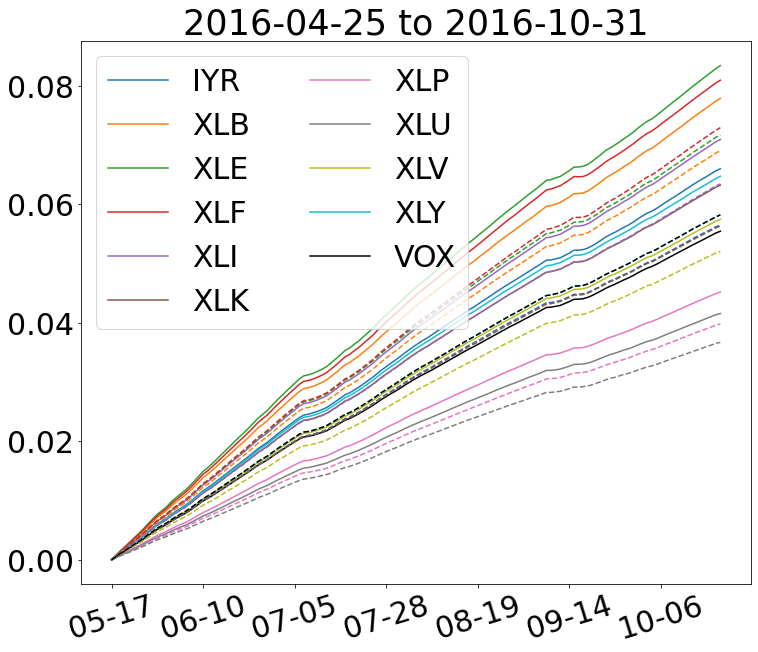}
        \includegraphics[width = \x, height = \y]{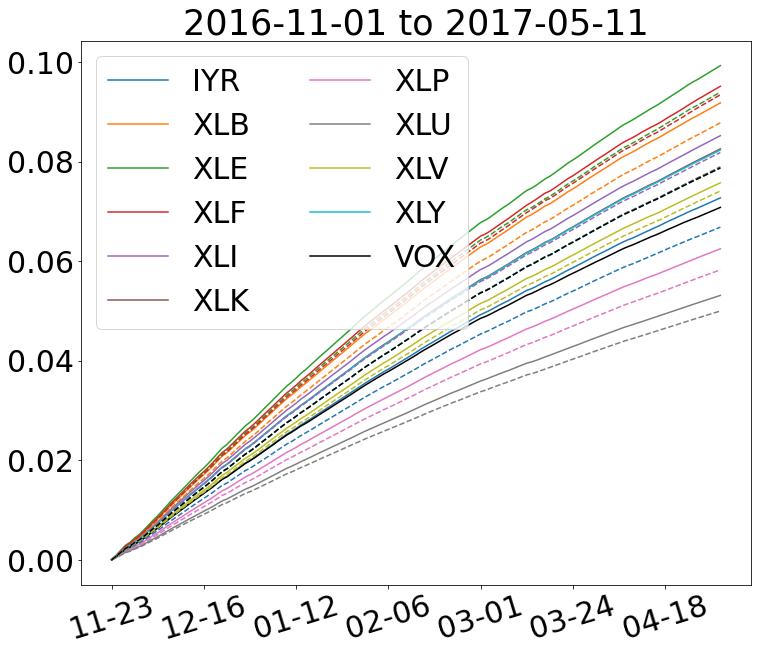}
        \includegraphics[width = \x, height = \y]{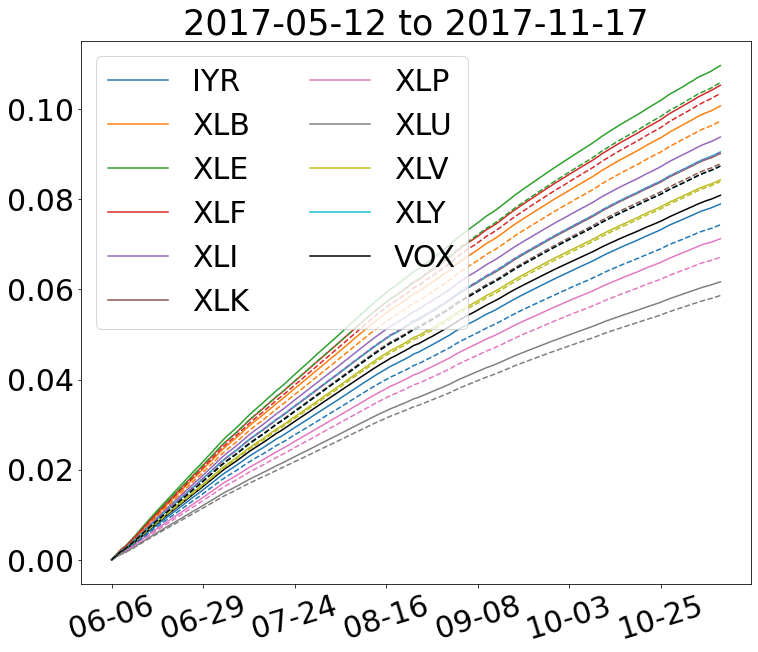}
        \includegraphics[width = \x, height = \y]{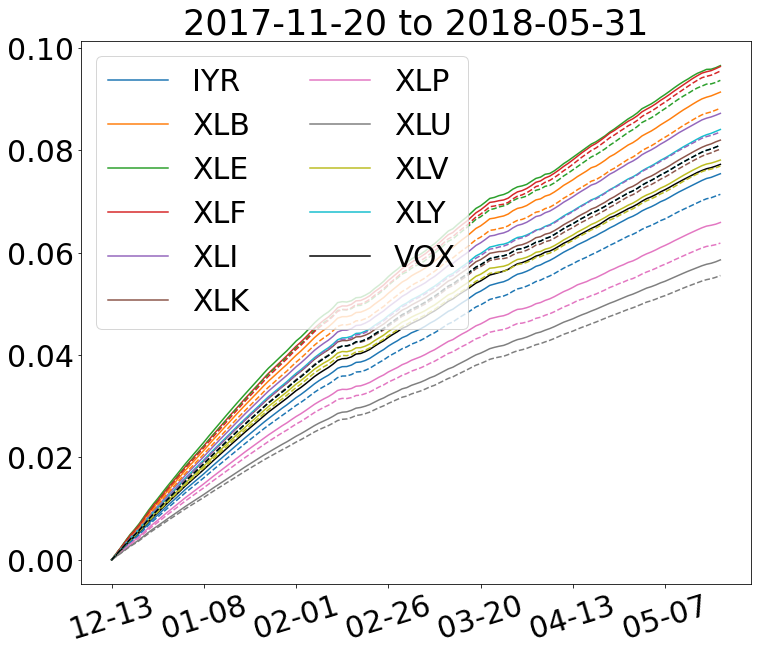}
        \includegraphics[width = \x, height = \y]{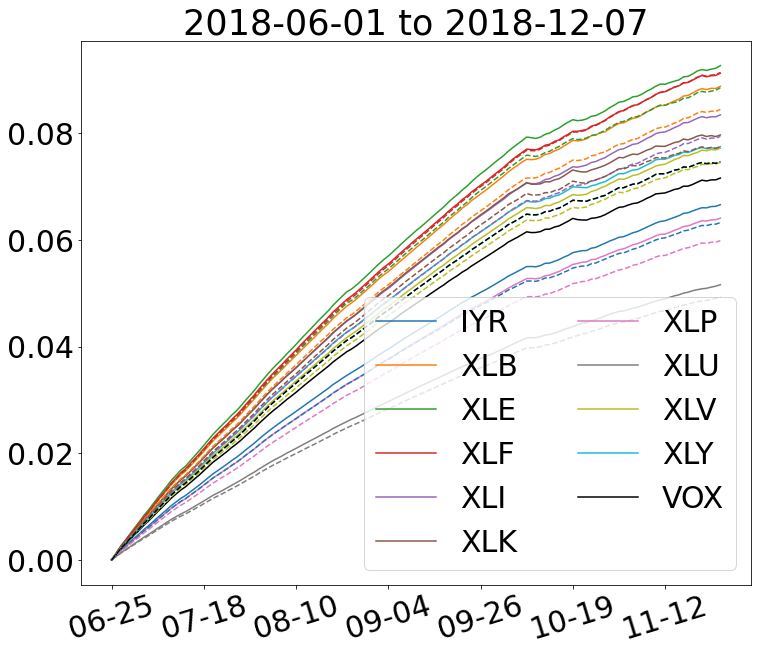}
        \includegraphics[width = \x, height = \y]{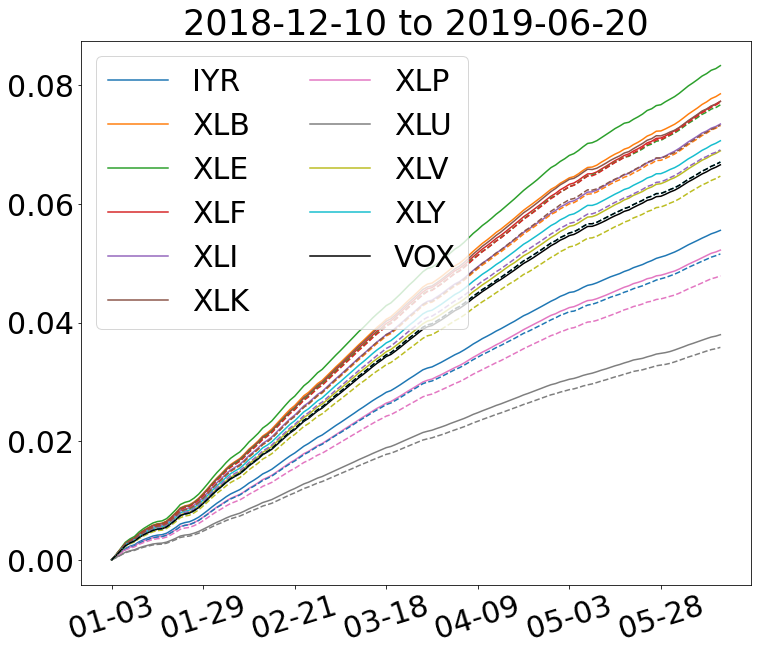}
        \includegraphics[width = \x, height = \y]{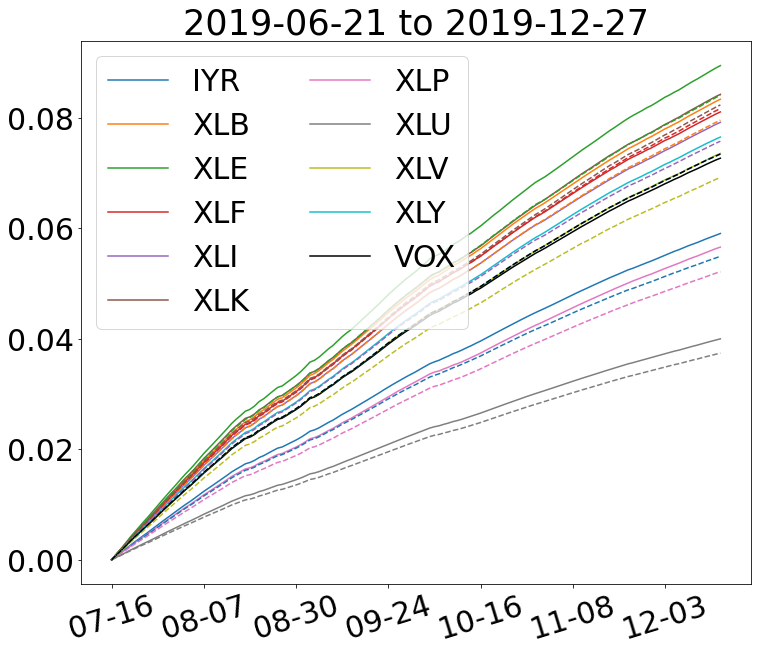}
        \caption{Portfolio value of the RL strategy (solid) and the myopic strategy (dot) with time. The X-axis shows the date (month - day).  Y-axis shows the portfolio.}
        \label{historicalportfolio}
\end{figure*}

\subsection{The Objective Value $V$ on Historical Data} The main goal of this paper is to maximize the objective value function \eqref{eq:valueFunction}, we, therefore, ran the RL and myopic algorithms on the 10 folds historical market data and calculated the relevant terms in \eqref{eq:valueFunction}. Specifically, we show the average results of $V$, total transaction cost $\sum_{t=0}^T \delta^t \frac{\epsilon q(\vec S_t,P_t)}{2}a_t^\top \Psi_t a_t$, and risk penalty $\sum_{t=0}^T \delta^t \frac{\gamma}{2}X_t^\top P_t X_t$ with different $\epsilon$.
We observed that when  $\epsilon$ is relatively small (e.g., $0.00001$), our RL algorithm considers an aggressive investment to maximize \eqref{eq:valueFunction} (i.e., increasing the inventory $X_t$ quicker and thus taking a relatively higher total transaction cost and risk penalty). When $\epsilon$ is relatively large (e.g., $0.01$), our RL algorithm adopts a conservative investment to maximize \eqref{eq:valueFunction} (i.e., increasing $X_t$ slower and thus taking a relatively lower total transaction cost and risk penalty). 
Fig.~\ref{fig:real_x} shows how the inventory of IYR ETF  changes for different $\epsilon$ by using the RL and myopic strategies. The inventories of other ETFs show similar behavior.  From the figure, when $\epsilon=0.01$, our RL tends to slowly increase $X_t$. When $\epsilon=0.00001$, our RL tends to increase $X_t$ quickly. The explanation for such behavior is that when $\epsilon$ is small (meaning that the transaction cost can be negligible), $f(a_t,X_t,\vec S_t,P_t)$ in  \eqref{eq:valueFunction} can be approximated as:
\begin{align}
\label{eq:approx-valueFunction}
f(a_t,X_t,\vec S_t,P_t)&\approx \underbrace{\mu^\top \Psi_t X_t}_{\hbox{investment return}} -\underbrace{\frac{\gamma}{2}X_t^\top P_t X_t}_{\hbox{risk penalty}}
\end{align} which has an equilibrium point at $X_t = \frac{1}{\gamma} P_t^{-1} \mu^\top \Psi_t $ for every $t$. Therefore, to maximize \eqref{eq:valueFunction} with the approximation \eqref{eq:approx-valueFunction}, our RL algorithm needs to increase $X_t$ to the equilibrium point within a short time and thus performs an aggressive investment. If $\epsilon$ is large (meaning that the transaction cost cannot be negligible),  to avoid high transaction costs, our RL can only slowly increase $X_t$. We show the average results of final objective-function value $V$, total transaction costs, and total risk penalty over the 10 folds historical market data in Table~\ref{tab:small_eps} for different $\epsilon$. The total transaction cost and risk penalty for $\epsilon=0.00001$ are higher than the total transaction cost and risk penalty for $\epsilon=0.01$, which confirms our explanation. 

\begin{figure}
    \centering
    \includegraphics[scale=0.5]{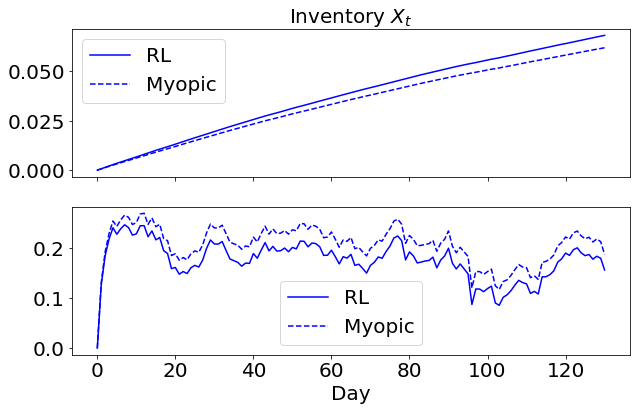}
    \caption{Inventory of IYR ETF. Top: $\epsilon=0.01$. Bottom: $\epsilon=0.00001$.}
    \label{fig:real_x}
\end{figure}

Our RL algorithm outperforms the myopic strategy for both aggressive and conservative investment cases by showing a higher  $V$ according to Table~\ref{tab:small_eps}. For the aggressive investment case ($\epsilon=0.00001$), our RL algorithm behaves less aggressively than the myopic strategy by having a lower total transaction cost and risk penalty. On the other hand, for the conservative investment case ($\epsilon=0.01$), our RL algorithm behaves more aggressively than the myopic strategy by taking a higher total transaction cost and risk penalty. However, for both cases, our RL algorithm  shows a higher value for the objective $V$ than the myopic strategy, as shown in Table~\ref{tab:small_eps}. 
\begin{table}
  \caption{Total Value $V$, Transaction Costs, and Risk Penalties for Different $\epsilon$.}
    \centering
    \begin{tabular}{|c|cccc|}
    \hline
   &  \multicolumn{4}{c|}{$V$}   \\
   $\epsilon$ &  RL & Myopic  & Difference & Difference / RL   \\
   \hline
   $0.01$ & 0.0234&	0.0225&	0.0009&	3.85\%	 \\
   $0.00001$ & 0.0885&	0.0782&	0.0103&	11.64\% 	\\
   \hline
 &  \multicolumn{4}{c|}{Total Transaction Cost} \\
   &  RL & Myopic  & Difference & Difference / RL \\
    \hline
  $0.01$   & 0.0013&	0.0011&	0.0002&	15.38\% \\
  $0.00001$ &   0.2562&	0.2600&	-0.0038&	-1.50\% \\
  \hline
   &  \multicolumn{4}{c|}{Total Risk Penalty} \\
  &   RL & Myopic  & Difference & Difference / RL \\
     \hline
 $0.01$  & 0.0012&	0.0011&	0.0001&	8.33\%  \\
   $0.00001$ &  0.2399&	0.2436&	-0.0037&	-1.54\% \\
   \hline
    \end{tabular}
    \label{tab:small_eps}
\end{table}

Fig.~\ref{fig:real} shows the average performance of our RL and the myopic strategies with time on the 10-fold historical market datasets.  The average wealth returns using  RL and myopic strategies are close. However, our RL approach attains a higher (i.e., better) value for the objective than the myopic strategy, meaning that our RL achieves a better trade-off between investment return, transaction cost, and risk penalty. Therefore, our RL outperforms the myopic strategy. 

\begin{figure}
    \centering
    \includegraphics[scale=0.35]{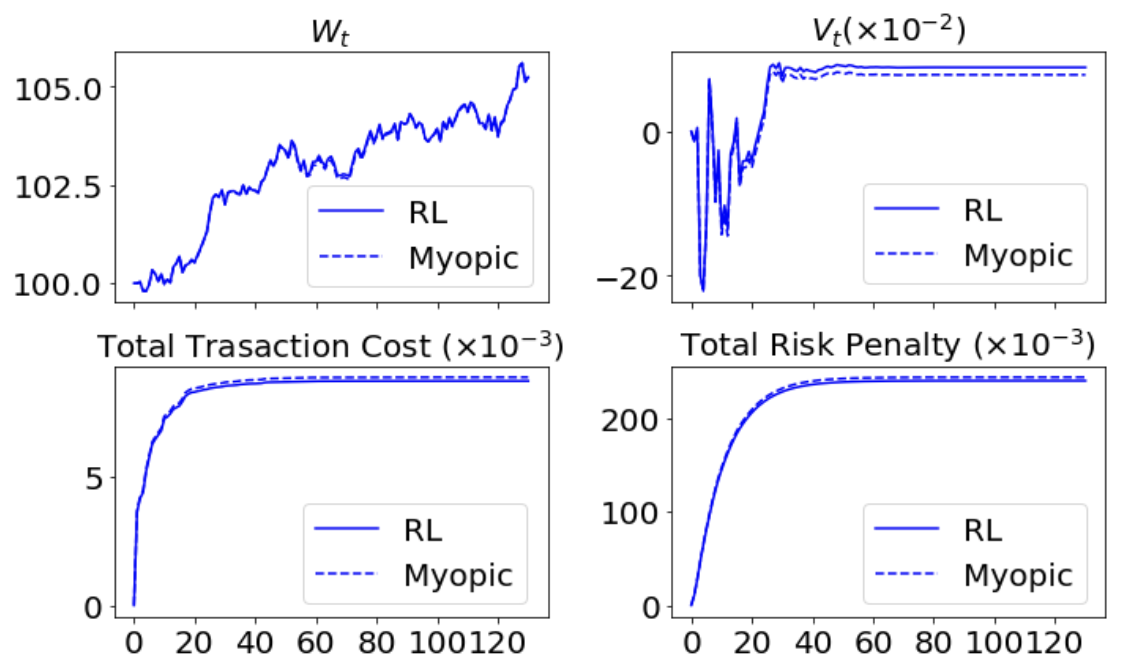}
    \includegraphics[scale=0.35]{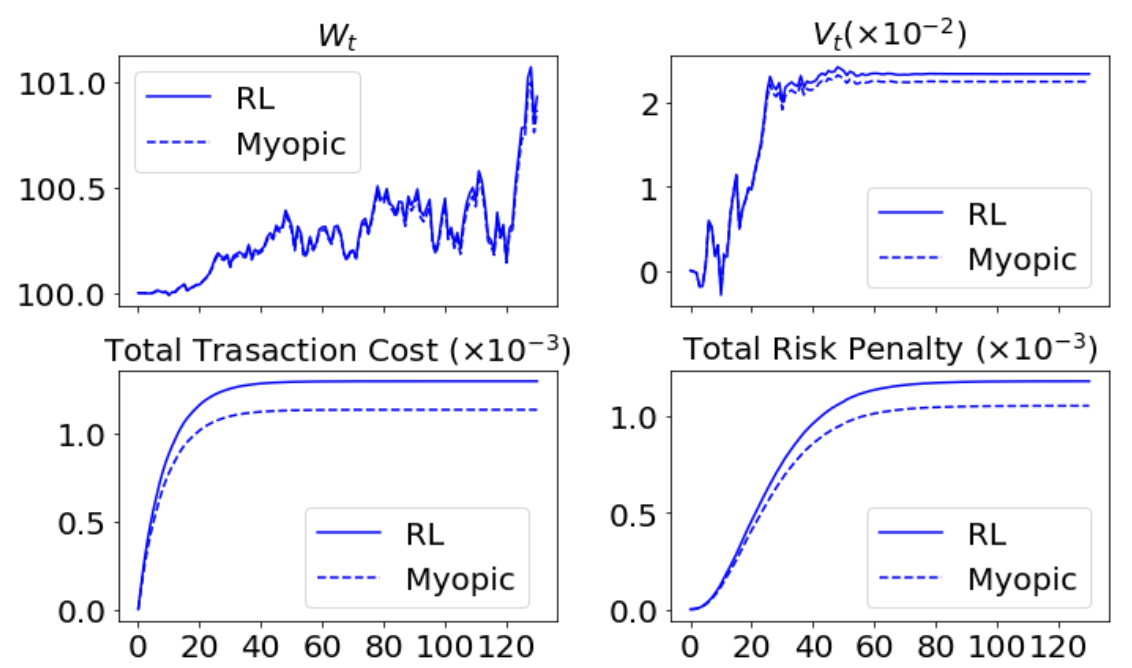}
    \caption{Results on real data. Top: $\epsilon=0.00001$. Bottom: $\epsilon=0.01$.}
    \label{fig:real}
\end{figure}

\subsection{Discussion of Our RL Approach}
From the experimental results, we have observed that our RL approach has a higher objective value V given in \eqref{eq:valueFunction} (i.e., the main goal of this paper) than the myopic strategy on the historical data. In other words, our RL approach balances the wealth return, risk penalty, and transaction cost better than the myopic by showing a 3.85\% better for $\epsilon=0.01$ and 11.64\% better for $\epsilon=0.00001$. On synthetic data, the RL method returns an extra 1-2\% additional annual percentage return over the myopic (out of a total of roughly 10\% annual return). This demonstrates that our method outperforms the myopic, as this extra 1-2\% annually is significant when measuring the growth of assets. For real data, we are able to show consistent over-performance of the RL method by about 11bps. This over-performance certainly makes it evident  RL's  improved out-of-sample performance. Our RL approach also shows a faster investment speed with working in real time. Therefore, our approach is a valid effective method.

\section{Conclusion}
We have implemented a reinforcement learning algorithm to solve the mean-variance preferences of \eqref{eq:valueFunction} with an MGARCH model and transaction costs of order $\epsilon$, where $\epsilon>0$ is a small parameter. Our method addresses issues of algorithm convergence and computation time by using an $\epsilon$ expansion to guide the network to the solution. While similar methods use double (dueling) networks to stabilize convergence, we only need to use one network with our expansion. Our method is stable and can work in real time. The resulting portfolios show good performance in simulated tests.  

\appendix
\section{Proofs}
{\bf Proof of Proposition \ref{prop:deltaEps}}:
\label{app:DeltaProof}
    We apply the Sherman-Woodbury-Morrison formula,
    \begin{align}
        \label{eq:SWM}
        \left(I+\widetilde P_t\Psi_t^{-1}\right)^{-1}&= \Psi_t\widetilde P_t^{-1} -\Psi_t\widetilde P_t^{-1}\left(I+\Psi_t\widetilde P_t^{-1}\right)^{-1}\Psi_t\widetilde P_t^{-1}\ .
    \end{align}
    The first term in \eqref{eq:SWM} is bounded as follows,
    \begin{align*}
        \|\Psi_t\widetilde P_t^{-1}\| &= \frac{\epsilon }{\gamma} q(\vec S_t,P_t)\|P_t^{-1}\Psi_t\|  \leq \frac{\epsilon }{\gamma}\goodchi \|h\|_\infty\|P_t^{-1}\|\leq\frac{\epsilon \goodchi\|h\|_\infty}{\gamma\underbar s^2c}\ ,
    \end{align*}
    where we have bounded $\|P_t^{-1}\|$ in the following way,
        \[\|P_t^{-1}\|\leq \|\Psi_t^{-1}\|^2\cdot \|\Sigma_t^{-1}\|\leq \frac{1}{\underbar s^2\inf_{\|v\|=1}v^\top CC^\top v}=\frac{1}{\underbar s^2c}\ ,\]
    where $c>0$ is the lower bound defined in Condition \ref{cond:C}. The second term in \eqref{eq:SWM} is bounded as follows,
        \begin{align*}
            &\left\|\Psi_t\widetilde P_t^{-1}\left(I+\Psi_t\widetilde P_t^{-1}\right)^{-1}\Psi_t\widetilde P_t^{-1}\right\| \leq \left\|\Psi_t\widetilde P_t^{-1}\right\|^2\left\|\left(I+\Psi_t\widetilde P_t^{-1}\right)^{-1}\right\| \leq \left(\frac{\epsilon \goodchi\|h\|_\infty}{\gamma\underbar s^2c}\right)^2\left\|\left(I+\Psi_t\widetilde P_t^{-1}\right)^{-1}\right\|\ .
        \end{align*}
    Now, taking the norm of \eqref{eq:SWM}, applying a triangle inequality to the right-hand side, and inserting these bounds, we have
    \begin{align*}
        \left\|\left(I+\Psi_t\widetilde P_t^{-1}\right)^{-1}\right\| &\leq \frac{\epsilon \goodchi\|h\|_\infty}{\gamma\underbar s^2c}  +\left(\frac{\epsilon \goodchi\|h\|_\infty}{\gamma\underbar s^2c}\right)^2\left\|\left(I+\Psi_t\widetilde P_t^{-1}\right)^{-1}\right\|\ .
    \end{align*}
    If $\epsilon \goodchi\|h\|_\infty<\gamma\underbar s^2c$, then we can re-arrange to obtain the bound in \eqref{eq:smallEpsBound}. 
    
{\bf Proof of Theorem \ref{thm:convThm}}:
\label{app:convergenceProof}
The following lemma is useful to have when proving convergence in Theorem \ref{thm:convThm}.
\begin{lemma}
    \label{lemma:AB}
    Let $A$ and $B$ be two symmetric positive definite matrices. Then
    $\|(A+B)^{-1}\|\leq \min(\|A^{-1}\|,\|B^{-1}\|)$.
\end{lemma}
\begin{proof}
    By positive definiteness of each matrix, we have
    \begin{align*}
        &\inf_{\|v\|=1}v^\top(A+B)v \geq \inf_{\|v\|=1}v^\top Av+\inf_{\|v\|=1}v^\top Bv  =\frac{1}{\|A^{-1}\|}+\frac{1}{\|B^{-1}\|}\geq \frac{1}{\min(\|A^{-1}\|,\|B^{-1}\|)}\ .
    \end{align*}
    Thus, we have
    \begin{align*}
        \|(A+B)^{-1}\|&=\frac{1}{\inf_{\|v\|=1}v^\top(A+B)v}  \leq \min(\|A^{-1}\|,\|B^{-1}\|)\ ,
    \end{align*}
    which proves the second statement.
\end{proof}

From inspection of \eqref{eq:fullyForwardForm}, the usefulness of the following lemma should be evident:

\begin{lemma}
    \label{lemma:boundedMatrices}
    If we assume Condition \ref{cond:C}, Condition \ref{cond:S_LB} and Condition \ref{cond:q_bound}, then in \eqref{eq:fullyForwardForm} there are the following bounds in coefficients,
    \begin{align}
        \label{eq:matrixBounds}
        \left\|\frac{1}{\epsilon q(\vec S_{t},P_t)}\Psi_{t}^{-1}\right\|&\leq \frac{1}{\epsilon \underbar s}\\
        \nonumber
        \left\|\left(I + \widetilde P_t\Psi_t^{-1}\right)^{-1}P_t\right\|&\leq \frac{\epsilon }{\gamma} \goodchi\|h\|_\infty\ ,
    \end{align}
    where $\underbar s$ and $\|h\|_\infty$ are constants defined in Condition \ref{cond:S_LB} and $\goodchi$ is the bound defined in Condition \ref{cond:q_bound}.
\end{lemma}

\begin{proof}
~\newline
    \begin{enumerate}
        \item The first bound in \eqref{eq:matrixBounds} is a clear consequence of Condition \ref{cond:S_LB} and Condition \ref{cond:q_bound}. 
        \item To prove the second bound in \eqref{eq:matrixBounds}, we start as follows,
        \begin{align*}
            \left(I + \widetilde P_t\Psi_t^{-1}\right)^{-1}P_t&=\frac{\epsilon}{\gamma} q(\vec S_t,P_t)  \left(\frac{\epsilon }{\gamma}q(\vec S_t,P_t)P_t^{-1} + \Psi_t^{-1}\right)^{-1}\ ,
        \end{align*}
        where $P_t$ is invertible because we have assumed Condition \ref{cond:C} putting a lower bound on the eigenvalues. However, we don't have a lower bound on $\|P_t^{-1}\|$ because the eigenvalues of $P_t$ do not have a fixed upper bound, but from Condition \ref{cond:S_LB} we have $\|\Psi_t\|\leq \|h\|_\infty$, and thus we use Lemma \ref{lemma:AB} to get
        \begin{align*}
            &\frac{\epsilon }{\gamma} q(\vec S_t,P_t)\left\|\left(\frac{\epsilon }{\gamma} q(\vec S_t,P_t)P_t^{-1} + \Psi_t^{-1}\right)^{-1}\right\| \leq \frac{\epsilon }{\gamma} q(\vec S_t,P_t)\|h\|_\infty\ ,
        \end{align*}
        which proves the second bound in \eqref{eq:matrixBounds}.
    \end{enumerate}
\end{proof}

We are now ready to prove Theorem \ref{thm:convThm}. From \eqref{eq:fullyForwardForm} and using $\Omega$ defined in \eqref{eq:Delta_and_Omega}, we obtain the following inequality,
\begin{align}
        \nonumber
    &\|\lambda_t^{(k+1)} - \lambda_t^{(k)}\|
    \leq   \delta\|(I+\widetilde P_t\Psi_t^{-1})^{-1}\|\cdot\mathbb E_t\|\lambda_{t+1}^{(k+1)} - \lambda_{t+1}^{(k)}\|\\
    \nonumber
    &+\|\gamma(I+\widetilde P_t\Psi_t^{-1})^{-1}P_t\|\sum_{t' = 1}^{t-1}\frac{\|\Psi_{t'}^{-1}\|}{\epsilon q(\vec S_{t'},P_{t'})}\|\lambda_{t'}^{(k)} - \lambda_{t'}^{(k-1)}\|\\
    &\leq   \Delta(\epsilon)\mathbb E_t\|\lambda_{t+1}^{(k+1)} - \lambda_{t+1}^{(k)}\|+\Omega\sum_{t' = 1}^{t-1}\|\lambda_{t'}^{(k)} - \lambda_{t'}^{(k-1)}\|\ ,
    \label{eq:lambda_k}
\end{align} 
where from Condition \ref{cond:damping} we have $\left\|\delta\left(I+ \widetilde P_t\Psi_t^{-1}\right)^{-1}\right\|\leq \Delta(\epsilon)<1$ and from Lemma \ref{lemma:boundedMatrices} it follows that
\[\sup_{t,t'}\frac{\|\Psi_{t'}^{-1}\|}{\epsilon q(\vec S_{t'},P_{t'})}\left\|\left(I+ \widetilde P_t\Psi_t^{-1}\right)^{-1}P_t\right\|\leq \Omega\ .\] 
If we consider a finite-time version of the problem where $\lambda_{T+1}^{(k)} \equiv 0$ for all $k$ (i.e., no control on $X_t^{(k)}=X_T^{(k)}$ for $t>T$), then we can prove convergence of the iteration scheme. Denote the expected iteration error as
\[\mathcal E_t^{(k)} = \mathbb E\|\lambda_t^{(k)} - \lambda_t^{(k-1)}\|\ .\]
Denoting the vector of errors as $$\mathcal E_{1:T}^{(k)}=\Big(\mathcal E_T^{(k)},\mathcal E_{T-1}^{(k)},\dots,\mathcal E_1^{(k)}\Big)^\top,$$ the inequality in \eqref{eq:lambda_k} can be expressed with the following matrix/vector system,
\begin{align}
	\label{eq:error_bound}
	\mathcal E_{1:T}^{(k+1)}&\leq H_T\mathcal E_{1:T}^{(k+1)}+N_T\mathcal E_{1:T}^{(k)}
\end{align}
where $H_T$ the $T\times T$ nilpotent matrix whose entry at row $i$ column $j$ is
\[
H_T^{ij}=\Delta(\epsilon)\times
\begin{cases}
    \hbox{$1$ if $i=j+1$}\\
    \hbox{$0$ otherwise,}
\end{cases}
\]
which has norm $\|H_T^k\| = \Delta^k(\epsilon)$ for $k<T$ and $\|H_T^k\|=0$ for $k>T$, and where $N_T$ is the $T\times T$ nilpotent matrix whose entry at row $i$ column $j$ is
\[
N_T^{ij}=\Omega\times 
\begin{cases}
    \hbox{$1$ if $i<j$}\\
    \hbox{$0$ otherwise,}
\end{cases}
\]
for which $N_T^k = 0$ for $k\geq T$. Therefore, by combining the error bound in \eqref{eq:error_bound} we have
\begin{align*}
\|\mathcal E_{1:T}^{(k+1)}\| 
&\leq \left\|H_T\mathcal E_{1:T}^{(k+1)}\right\|+\left\|N_T\mathcal E_{1:T}^{(k)}\right\| \leq\Delta(\epsilon)\left\|\mathcal E_{1:T}^{(k+1)}\right\|+\left\|N_T\mathcal E_{1:T}^{(k)}\right\| \ ,
\end{align*}
and after rearranging and taking $k$-many iterations, we have
\begin{align}
    \label{eq:iterationBound}
    \|\mathcal E_{1:T}^{(k+1)}\|& \leq \left(\frac{1}{1-\Delta(\epsilon)}\right)^k\|N_T^k\mathcal E_{1:T}^{(1)}\| \leq \left(\frac{1}{1-\Delta(\epsilon)}\right)^k\|N_T^k\|\|\mathcal E_{1:T}^{(1)}\|\ .
\end{align}
Now we take a moment to derive a bound on the normal of $N_T$: comparing sums with integrals we see the following,

\begin{align*}
    \|N_T^ke_T\|_1 &\leq \Omega^k\sum_{t_1=1}^{T-1}\sum_{t_2=1}^{t_1 - 1}\dots\sum_{t_k=1}^{t_{k-1}-1}1  \leq \Omega^k\int_0^Tdt_1\int_0^{t_1}dt_2\dots \int_0^{t_{k-1}}dt_k = \frac{\Omega^kT^k}{k!}\ ,
\end{align*}
for $k<T$ where $e_T\in\mathbb R^T$ is the $T^{th}$ canonical basis vector, and this bound gives us the general bound 
\begin{equation}
	\label{eq:iterated_integrals}
	\|N_T^k\|\leq\sqrt{T}\|N_T^ke_T\|_1\leq \frac{\Omega^kT^{k+1/2}}{k!}\mathbbm{1}_{k<T}\qquad\forall k\ .
\end{equation}
We apply the bound in \eqref{eq:iterated_integrals} to the right-hand side in \eqref{eq:iterationBound}, and we find that $\|\mathcal E_{1:T}^{(k+1)}\|=0$ for all $k\geq T$ thus proving convergence to a unique fixed point, hence proving the statement of Theorem \ref{thm:convThm}.

{\bf Proof of Theorem \ref{thm:NN_fixed_point}}:
\label{app:NNconvergenceProof}
Given the compact set $\mathbf K$ from \eqref{eq:UI}, by the universal approximation theorem (\cite{cybenko1989approximation}), for each iteration in \eqref{eq:implicitScheme} there is a NN for which parameters can be chosen such that
\begin{align*}
    &\sup_{t\leq T}\mathbb E\Big[\left\|\lambda(t,X_{t-1}^{(k)},\vec S_t,P_t;\theta^{(k)})- Y_t^{(k-1)}(\theta^{(k)})\right\|  \mathbbm{1}_{\{(X_{t-1}^{(k)},\vec S_t,P_t)\in \mathbf K\}}\Big]\leq  \varepsilon^{(k)}\ ,
\end{align*}
where $\varepsilon^{(k)}$ is the arbitrarily small error that can be decreased by increasing the NN's hyperparameters. In the same manner as the process given by \eqref{eq:iterationScheme}, we define the iterated process to be
\begin{align*}
    \lambda_t^{(k)}& = \lambda(t,X_{t-1}^{(k)},\vec S_t,P_t,\theta^{(k)})\ .
\end{align*}
Denote the error as
\[\mathcal E_t^{(k)}=\mathbb E\left\|\lambda_t^{(k)}-\lambda_t^*\right\|\ ,\]
where $\lambda_t^*$ is the limiting fixed-point of \eqref{eq:iterationScheme}. Using the pair of iterative equations in \eqref{eq:iterationScheme}, we can include the NN error to obtain the following recursive bound,
\begin{align*}
    &\mathcal E_t^{(k)}\leq  \frac{\mathbb E\left[\|\gamma(I+\widetilde P_t\Psi_t^{-1})^{-1}P_t\|\cdot\|X_{t-1}^{(k)}-X_{t-1}^*\|\right]+(\varepsilon^{(k)}+2\eta)} {1-\Delta(\epsilon)}\ .
\end{align*}
Then, proceeding similarly to the proof of Theorem \ref{thm:convThm} in Appendix, we denote the vector of all errors as $\mathcal E_{1:T}^{(k)}=\Big(\mathcal E_T^{(k)},\mathcal E_{T-1}^{(k)},\dots,\mathcal E_1^{(k)}\Big)^\top$, and like we had in \eqref{eq:error_bound}, a bound on these errors can be expressed with the following matrix/vector system,
\begin{align*}
    \mathcal E_{1:T}^{(k+1)}&\leq \frac{1}{1-\Delta(\epsilon)}\left(N_T\mathcal E_{1:T}^{(k)}+\sup_\ell(\varepsilon^{(\ell)}+2\eta)\mathbf 1\right)\\
    &\leq \left(\frac{1}{1-\Delta(\epsilon)}\right)^kN_T^k\mathcal E_{1:T}^{(1)}  +\sup_\ell(\varepsilon^{(\ell)}+2\eta)\sum_{i=0}^{k-1}\left(\frac{1}{1-\Delta(\epsilon)}\right)^{i+1}N_T^i\mathbf 1\ ,
\end{align*}
where $H_T$ and $N_T$ are the same matrices that were defined in Appendix, and where $\mathbf 1$ denotes the vector in $\mathbb R^T$ of all 1's.
Thus,
\begin{align*}
    &\|\mathcal E_{1:T}^{(k+1)}\|
    \leq \left(\frac{1}{1-\Delta(\epsilon)}\right)^k\|N_T^k\mathcal E_{1:T}^{(1)}\|  +\sup_\ell(\varepsilon^{(\ell)}+2\eta)\sum_{i=0}^{k-1}\left(\frac{1}{1-\Delta(\epsilon)}\right)^{i+1}\|N_T^i\mathbf 1\|\\
    &\leq \left(\frac{1}{1-\Delta(\epsilon)}\right)^k\|N_T^k\|\|\mathcal E_{1:T}^{(1)}\|  +\sqrt{T}\left(\sup_\ell\varepsilon^{(\ell)}+2\eta\right)\sum_{i=0}^{k-1}\left(\frac{1}{1-\Delta(\epsilon)}\right)^{i+1}\|N_T^i\|\\
    &= \left(\frac{1}{1-\Delta(\epsilon)}\right)^k\frac{\Omega^kT^{k+1/2}}{k!}\mathbbm{1}_{k<T}\|\mathcal E_{1:T}^{(1)}\|+ \sqrt{T}\left(\sup_\ell\varepsilon^{(\ell)}+2\eta\right)\sum_{i=0}^{k-1}\left(\frac{1}{1-\Delta(\epsilon)}\right)^{i+1}\frac{\Omega^iT^{i+1/2}}{i!}\mathbbm{1}_{i<T}\ ,
\end{align*}
 for large $k\geq T$ we have

\[\|\mathcal E_{1:T}^{(k+1)}\|=\mathcal O\left(\left(\sup_\ell\varepsilon^{(\ell)}+2\eta\right)\exp\Big(\frac{\Omega T}{1-\Delta(\epsilon)}\Big)\right)\ ,\]
which is the statement of Theorem \ref{thm:NN_fixed_point}.


{\bf Proof of Proposition \ref{prop:expansionAccuracy}}:
\label{app:smallEpsilonProof}
Denote the state process from the order-$\epsilon$ approximation as
\[X_{t}^{[1,2]}=X_{t-1}^{[1,2]}-\sum_{t'=1}^{t}\frac{1}{q(\vec S_{t'}P_{t'})}\Psi_{t'}^{-1}(\tilde\lambda_{t'}^{[0]}+\epsilon\widetilde\lambda_t^{[1]})\ ,\]
for $t=1,2,3,\dots, T$. Inserting the expansion into \eqref{eq:LambdatildeLambdaSystem} results in the order-$\epsilon$ approximation error
    \begin{align*}
        & \widetilde\lambda_t^{[0]}+\epsilon \widetilde\lambda_t^{[1]} -\left(\epsilon I+\frac{\gamma}{q(\vec S_t,P_t)} P_t\Psi_t^{-1}\right)^{-1}  \left(\delta \epsilon\mathbb E_t\left[\widetilde\lambda_{t+1}^{[0]}+\epsilon \widetilde\lambda_{t+1}^{[1]}\right]-\gamma P_t\left(\frac{1}{\gamma}P_t^{-1}\Psi_t\mu - X_{t-1}^{[1,2]}\right)\right)\\
        &=-\left(\epsilon I+\frac{\gamma}{q(\vec S_t,P_t)} P_t\Psi_t^{-1}\right)^{-1}\left(\delta \epsilon^2\mathbb E_t \widetilde\lambda_{t+1}^{[1]}\right)\ ,
    \end{align*}
    where the equality uses the expressions of for $\widetilde\lambda_t^{[0]}$ and $\widetilde \lambda_t^{[1]}$ given in \eqref{eq:expansionTerms}. Using Lemma \ref{lemma:boundedMatrices}, we have
    \[\left\|\left(\epsilon I+\frac{\gamma}{q(\vec S_t,P_t)} P_t\Psi_t^{-1}\right)^{-1}\right\|=\frac{\Delta(\epsilon)}{\epsilon}\ .\]
    Therefore the differential error given above is of order $\mathcal O\left(\epsilon\Delta(\epsilon)\mathbb E\sup_{t\leq T}\|\widetilde\lambda_t^{[1]}\| \right)$. Thus, similar to Theorem \ref{thm:NN_fixed_point}, we have a big-oh bound on error,
    \begin{align*}
        &\sup_{t\leq T}\mathbb E\left\|\widetilde\lambda_t^{[0]}+\epsilon\widetilde\lambda_t^{[1]}-\widetilde\lambda_t^*\right\|=  \mathcal O\left(\epsilon\Delta(\epsilon)\exp\Big(\frac{\Omega T}{1-\Delta(\epsilon)}\Big)\mathbb E\sup_{t\leq T}\|\widetilde\lambda_t^{[1]}\|\right)\ ,
    \end{align*} 
    where $\widetilde\lambda^*$ denotes the solution from \eqref{eq:XtildeLambdaSystem} and \eqref{eq:LambdatildeLambdaSystem}.

\begin{table}
\caption{NN Architecture.}
\begin{center}
\begin{tabular}{cccc}
\toprule
\textbf{Layer} & \textbf{Type} & \textbf{Size} & \textbf{Activation}   \\
\hline
Input & Linear & 143 $\to$ 400 & tanh \\ 
Hidden & Linear & 400 $\to$ 400 & tanh \\ 
Hidden & Linear & 400 $\to$ 400 & tanh \\ 
Hidden & Linear & 400 $\to$ 400 & tanh \\ 
Hidden & Linear & 400 $\to$ 400 & tanh \\ 
Output & Linear & 400 $\to$ 11 & tanh \\ 
\bottomrule
\end{tabular}
\end{center} 
\label{table:architecture}
\end{table}

\bibliographystyle{IEEEtranN}
\bibliography{arxiv}

\clearpage

\end{document}